\documentclass[12pt, draftclsnofoot, onecolumn]{IEEEtran} 
\usepackage{amssymb,amsmath}
\usepackage{amsthm}

\newtheorem{lemma}{Lemma}
\newtheorem{theorem}{Theorem}
\newtheorem{corollary}{Corollary}
\newtheorem{remark}{Remark}
\usepackage[usenames, dvipsnames]{color}

\usepackage{array}
\newcolumntype{L}[1]{>{\raggedright\let\newline\\\arraybackslash\hspace{0pt}}m{#1}}
\newcolumntype{C}[1]{>{\centering\let\newline\\\arraybackslash\hspace{0pt}}m{#1}}
\newcolumntype{R}[1]{>{\raggedleft\let\newline\\\arraybackslash\hspace{0pt}}m{#1}}

\usepackage{cite}
\usepackage{graphicx,subfigure}
\usepackage{psfrag}
\usepackage{url}
\usepackage[latin1]{inputenc}
\usepackage{stackengine}
\def\delequal{\mathrel{\ensurestackMath{\stackon[1pt]{=}{\scriptscriptstyle\Delta}}}}

\begin{document}
	
	\title{Aggregation and Resource Scheduling in Machine-type Communication Networks: A Stochastic Geometry Approach}
	\author{
		\IEEEauthorblockN{Onel L. Alcaraz López, \IEEEmembership{Student Member, IEEE}, %
			Hirley Alves, \IEEEmembership{Member, IEEE},
			Pedro H. J. Nardelli, 
			Matti Latva-aho, \IEEEmembership{Senior Member, IEEE}
		}		%
		
		\thanks{Onel L. Alcaraz López, Hirley Alves, Matti Latva-aho are with the Centre for Wireless Communications (CWC), University of Oulu, Finland.\{onel.alcarazlopez,hirley.alves,pedro.nardelli,matti.latva-aho\}@oulu.fi}
		
		\thanks{Pedro H. J. Nardelli is with Laboratory of Control Engineering and Digital Systems, Lappeenranta University of Technology, Finland. pedro.nardelli@lut.fi}

		\thanks{This work is partially supported by Academy of Finland (Aka) (Grants n.303532, n.307492), SRC/Aka (n. 292854), and by the Finnish Funding Agency for Technology and Innovation (Tekes), Bittium Wireless, Keysight Technologies Finland, Kyynel, MediaTek Wireless, Nokia Solutions and Networks.}
	}

\maketitle

\begin{abstract}
Data aggregation is a promising approach to enable massive machine-type communication (mMTC). This paper focuses on the aggregation phase where a massive number of machine-type devices (MTDs) transmit to aggregators. By using non-orthogonal multiple access (NOMA) principles, we allow several MTDs to share the same orthogonal channel in our proposed hybrid access scheme. We develop an analytical framework based on stochastic geometry to investigate the system performance in terms of average success probability and average number of simultaneously served MTDs, under imperfect successive interference cancellation (SIC) at the aggregators, for two scheduling schemes: random resource scheduling (RRS) and channel-aware resource scheduling (CRS).  We identify the power constraints on the MTDs sharing the same channel to attain a fair coexistence with purely orthogonal multiple access (OMA) setups. Then, power control coefficients are found, so that these MTDs perform with similar reliability. We show that under high access demand, the hybrid scheme with CRS outperforms the OMA setup by simultaneously serving more MTDs with reduced power consumption.
\end{abstract}
\section{Introduction}
Machine-type Communication (MTC) is typically cited as an integral use case for fifth generation (5G) cellular networks, where fully automatic data generation, exchange, processing and actuation among intelligent machines, are on the agenda. With the rapid penetration of embedded devices, MTC is becoming the dominant communication paradigm for a wide range of emerging smart services including healthcare, manufacturing, utilities, consumer goods and transportation. Specifically, massive MTC (mMTC), as the name suggests, is about massive  access by a large number of devices, that is, about providing wireless connectivity to enormous number of often low-complexity low-power machine-type devices (MTDs). Of course, the access is a major concern and a flouring research area. Different methods have been proposed and addressed to provide more efficient access, e.g., access class barring \cite{ETSI1}, prioritized random access \cite{Lin.2014}, backoff adjustment scheme \cite{Yang.2012}, delay-estimation based random
access \cite{Hossain.2016}, distributed queuing \cite{Laya.2016}. Another promising way to deal with the massive connection problem comes from the concept of data aggregation. Instead of being directly connected to the access network, e.g., by equipping MTDs with own Subscriber Identity Module (SIM) card to have cellular connectivity, the idea is that MTDs may organize themselves locally, creating MTC area networks and exploiting short-range technologies. These MTC area networks may then connect to the core networks through MTC gateways or data aggregators. This is a key solution strategy to collect, process, and communicate data in MTC use cases with static devices, especially if the locations of the devices are known, such as smart utility meters \cite{Nardelli2015} or video surveillance cameras \cite{Dawy.2017}. The traffic from MTDs is first transmitted to the designed data aggregator, while the aggregator then relays the collected packets to the core network \cite{Kim.2017}, e.g., the base station (BS), thus, reducing the congestion and the power consumption at the MTD side \cite{Guo.2017}.

When aggregating such huge number of MTDs, the density of the aggregators, although considerable smaller than the density of the MTDs, obviously will still be large, and the interference coming from the devices sharing the same resource could be significant. There is limited literature that  characterizes the interference in mMTC with data aggregation. The analysis in \cite{Kwon.2013} shows how the wireless channel, transmit power, and random deployment of data collectors affect the signal-to-interference-plus-noise ratio (SINR) distribution in random access networks with randomly deployed sensor nodes, and for some special cases, a simple form of signal-to-noise ratio (SNR) or SINR distribution is found. An energy-efficient data aggregation scheme for a hierarchical MTC network is proposed in \cite{Malak.2016}, while authors develop a coverage probability-based optimal data aggregation scheme for MTDs to minimize the energy density of the network.
In \cite{Khoshkholgh.2015}, a theoretical and numerical framework, which aims to assess, model and characterize the network energy consumption profile, is presented by exploiting the stochastic geometry tools. 
By incorporating some intelligence in the aggregator, the network performance improves, as shown in \cite{Guo.2017,Chang.2012,Gotsis.2012,Hamdoun.2015,Kumar.2016}  
for resource scheduling strategies. Among them, only \cite{Guo.2017} considers a more realistic scenario with a multi-cell network, hence the inter-cell interference, which is a critical issue, is taken into account. 
The authors introduce a tractable two-phase network model for mMTC, where MTDs first transmit to their serving aggregators (aggregation phase) and then, the aggregated data is delivered to BSs (relaying phase). Key metrics such as the MTD success probability, average number of successful MTDs and probability of successful channel utilization are investigated for two resource scheduling schemes RRS,  where aggregators randomly allocate the limited resources to MTDs, and CRS, where aggregators allocate resources to the MTDs having better channel conditions. Authors show that, compared to the CRS scheme, the RRS scheme can achieve similar performance as long as the resources in the aggregation phase are not very limited.

In its turn, non-orthogonal multiple access (NOMA) has recently attracted a lot of attention as a promising technology for the coming 5G networks to significantly improve the spectral efficiency of mobile communication networks, and/or to meet the demand of massive connectivity. 
Authors in \cite{Shirvanimoghaddam.2016} present NOMA as a promising solution to support a massive number of MTDs in cellular
networks, while tackling the main practical challenges and future research directions. Actually, NOMA is considered in \cite{Miao.2016,Shirvanimoghaddam.2017_2} when serving a massive number of devices in cellular-based MTC networks. Energy-efficient clustering and medium access control are investigated in \cite{Miao.2016} to minimize device energy consumption and prolong network battery lifetime, while authors in \cite{Shirvanimoghaddam.2017_2} study the throughput and energy efficiency of the NOMA scenario with a random packet arrival model and derive the stability condition for the system to guarantee the performance. 
The key idea of NOMA is to exploit the power domain for multiple access, such that multiple users can be multiplexed at different power levels but at the same time/frequency/code. SIC is utilized to separate superimposed messages at the receiver side \cite{Saito.2013_2}. 
In general NOMA can be applied on both, downlink and uplink.
A downlink NOMA transmission system is studied in \cite{Ding.2014}, where the authors show that the outage performance of NOMA depends critically on the choices of the targeted data rates and allocated power in scenarios with randomly deployed users; while a dynamic power allocation scheme is proposed in \cite{Yang.2016} to downlink and uplink NOMA scenarios with two users with various quality of service requirements. A NOMA scheme for uplink that allows more than one user to share the same subcarrier without any coding/spreading redundancy is proposed in \cite{Imari.2014}, while the authors establish an upper limit on the number of users per subcarrier to control the receiver complexity. Also in the uplink, authors in \cite{Abbas.2017} consider a massive uncoordinated NOMA scheme where devices have strict latency requirements and no retransmission opportunities are available.
Finally, a good overview of NOMA, from its combination with multiple-input multiple-output (MIMO) technologies to cooperative NOMA, as well as the interplay between NOMA and cognitive radio, is provided in \cite{Ding.2017}.

Even though all the recent advances, only few papers focus on evaluating the performance of NOMA by using the stochastic geometry, except for \cite{Ding.2014,Zhang.2016,Abbas.2017,Zhang.2017}. However, the inter-cell interference, which is a pervasive problem in most of the existing wireless networks, is not explicitly considered in \cite{Ding.2014,Abbas.2017}, and neither do many other works on NOMA.  In contrast, authors in \cite{Zhang.2016} do consider the inter-cell interference when evaluating the performance of downlink NOMA on coverage probability and average achievable rate, as well as in \cite{Zhang.2017} for downlink and uplink NOMA scenarios. 
In fact, and to the best of our knowledge, \cite{Zhang.2017} is the only work analyzing NOMA by means of stochastic geometry in uplink mMTC scenarios for large-scale networks. However, even when authors deal with dense networks, aggregation architectures are not considered, and SIC is assumed perfect in the uplink. This paper aims at filling that gap by characterizing the interference, average success probability and average number of simultaneously served MTDs, for uplink mMTC in a large-scale cellular network system overlaid with data aggregators. We focus on the aggregation phase where the MTDs are allowed to share the same orthogonal channel, while the resource scheduling is implemented at the aggregator side as in \cite{Guo.2017} for OMA setups. The main contributions of this work can be listed as follows:
\begin{itemize}
	\item We introduce a hybrid access protocol, OMA-NOMA, for the aggregation phase of mMTC systems, while we develop a general analytical framework to investigate its performance in terms of average success probability and average number of simultaneously served MTDs. 	
	\item We extend the resource scheduling schemes RRS and CRS, proposed in \cite{Guo.2017} to deal with the limited resources, to scenarios where several MTDs are allowed to share the same orthogonal channel. As expected, CRS fits better to our setup since it allows performing adequate power control easily, while improving the system performance. Also, we find the power constraints on the MTDs sharing the same channel in order to attain a fair coexistence of our scheme with purely OMA setups, while power control coefficients are found too, so these MTDs can perform with similar reliability.
	 
	\item We show that, even when the hybrid scheme could lead to a less reliable system with greater chances of outages per MTD, e.g., due to the additional intra-cluster interference, the number of simultaneous active MTDs could be significantly improved for high access demand scenarios, as long as the success probability does not decrease so much. In that sense, our scheme aims at providing massive connectivity in scenarios with high access demand, which is not covered by traditional OMA setups. Additionally, the CRS scheme requires even a lower average power consumption per orthogonal channel and per MTD, than the OMA setup. 
	\item We attain approximated, yet accurate, expressions when analyzing the CRS scheme. Compared to the time-consuming Monte-Carlo simulations, even heavier for our hybrid scheme than for the purely OMA setup, our analytical derivations allow for fast computation. Also, imperfection when implementing SIC is incorporated in our proposed analytical framework.
\end{itemize}

The remainder of the paper is organized as follows. Section \ref{system} presents the system model and assumptions. Sections \ref{RRS} and \ref{CRS} discuss the RRS and CRS scheduling schemes for our hybrid access protocol, respectively. Section~\ref{PC} studies the power consumption of both schemes for OMA and the hybrid protocol.  Section \ref{results} presents the numerical results, and finally Section~\ref{conclusions} concludes the paper.
\newline\textbf{Notation:} $\mathbb{E}[\!\ \cdot\ \!]$ denotes expectation, while $\Pr(A)$ and  $\Pr(A|B)$ are the probability of event $A$, and $\Pr(A)$ conditioned on $B$, respectively. $|S|$ is the cardinality of set $S$, $||x||$ is the Euclidean norm of vector $x$ and $\mathrm{mod}(a,b)$ is the modulo operation. $_xF_y$ is the generalized hyper-geometric function \cite[Eq. (16.2.1)]{DLMF}, $\psi(x)$ is the digamma function \cite[Eq.~(6.3.1)]{Abramowitz.1972}, $\Gamma(x)$ is the gamma function, while $Q(a,x)$ and $\mathrm{Beta}(x,a,b)$ are the  regularized incomplete gamma function \cite[Eq. (8.2.4)]{DLMF} and incomplete beta function \cite[Eq.~(6.2.1)]{Abramowitz.1972}, respectively. $\lfloor \cdot \rfloor$ and $\lceil \cdot \rceil$ are the floor and ceiling functions, respectively. $\mathbf{i}=\sqrt{-1}$ and $\mathrm{Im}\{z\}$ is the imaginary part of $z\in\mathbb{C}$. $f_X(x)$ and $F_X(x)$ are the Probability Density Function (PDF) and Cumulative Distribution Function (CDF) of random variable (RV) $X$, respectively. {$X\sim\mathrm{Exp}(1)$ is an exponential distributed RV with unit mean, e.g., $f_X(x)=\exp(-x)$ and $F_X(x)=1-\exp(-x)$; while $Y\sim\mathrm{Poiss}(\bar{m})$ is a Poisson distributed RV with mean $\bar{m}$, e.g., $\Pr(Y\!=\!y)\!=\!\tfrac{1}{y!}\bar{m}^y\exp(-\bar{m})$ and $F_Y(y)\!=\!Q(y\!+\!1,\bar{m})$.}
\section{System Model and Assumptions}\label{system}
Consider a cellular network overlaid with data aggregators, which are spatially distributed according to a homogeneous Poisson point process (PPP), denoted $\Phi_a$ with density $\lambda_a$. Each aggregator, which could function as an ordinary cellular user for certain BS, serves multiple MTDs located nearby. Thus, the result is a cluster point process uniquely defined as:
\begin{align}\label{phi1}
\Phi\delequal\bigcup\limits_{\mathbf{w}\in\Phi_a}\mathbf{w}+\mathcal{B}^{\mathbf{w}},
\end{align}
where $\Phi_a$ is the parent PPP and $\mathcal{B}^{\mathbf{w}}$ denotes the offspring point process where each point at $\mathbf{s}\in\mathcal{B}^{\mathbf{w}}$ is i.i.d. around the cluster center $\mathbf{w}\in\Phi_a$ with distance distribution $f_{r_a}(r_a)=\tfrac{2r_a}{R_a^2}$, e.g., uniformly distributed in the disk region of radius $R_a$. $K\sim\mathrm{Poiss}(\bar{m})$ is the instantaneous number of MTDs requiring service in each aggregator, e.g., number of points in $\mathcal{B}^{\mathbf{w}}$, thus, the process is a Matérn cluster point process\footnote{A Matérn cluster point process is appropriated when considering either static or low mobility MTDs beign served by the aggregators. Some use cases for mMTC over cellular are: smart utility metering and industry automation \cite{Dawy.2017,Guo.2017}.}. Notice, that each MTD is associated with a single aggregator even when it could be located within the aggregation areas of several aggregators.

We focus on the uplink where the MTDs across the entire network are served through the same set of orthogonal channels, $\mathcal{N}$, available at each aggregator, with $|\mathcal{N}|=N$. Differently from \cite{Guo.2017}, here the same orthogonal channel could be used for more than one MTD. Some aggregators could be allocating one MTD per channel because the access demand is not so high, but some other aggregators could be allocating more MTDs per channel to face the increasing access demand, thus we propose and assess a hybrid OMA-NOMA multiple access scenario. Notice that the same aggregator could have channels operating with only one MTD, while others are operating with more. The maximum number of users per orthogonal channel is $L$, where $L=1$ reduces to the OMA scenario analyzed in \cite{Guo.2017}, which is used here as a benchmark. We focus on the $L=2$ setup\footnote{This is for analytical tractability, but notice that even when some existing results show that NOMA with more devices may provide a better performance gain \cite{Zhang.2016}, this may not be practical. The reason is that considering processing complexity for SIC receivers, especially when SIC error propagation is considered, 2-users NOMA is actually more practical in reality \cite{Liu.2016,Zhang.2017}.}, although some of our results hold for any $L\ge1$. Notice that for $L>1$ there is both: inter-cluster interference, e.g., interference from MTDs in the serving zone of other aggregators; and intra-cluster interference, e.g., interference from MTDs within the serving area of the same aggregator. After aggregating the MTDs' data, each aggregator relays the entire information to its associated BS. However, the focus of our current  work is on the aggregation phase and we assume that this phase occurs synchronously in all aggregators. Fig.\ref{Fig1} shows a snapshot of the considered network model. The silent MTDs are those out of the $N\cdot L$ available resources being used by the active MTDs.
\begin{figure}[t!]
	\centering
	\subfigure{\label{Fig1}\includegraphics[width=0.90\textwidth]{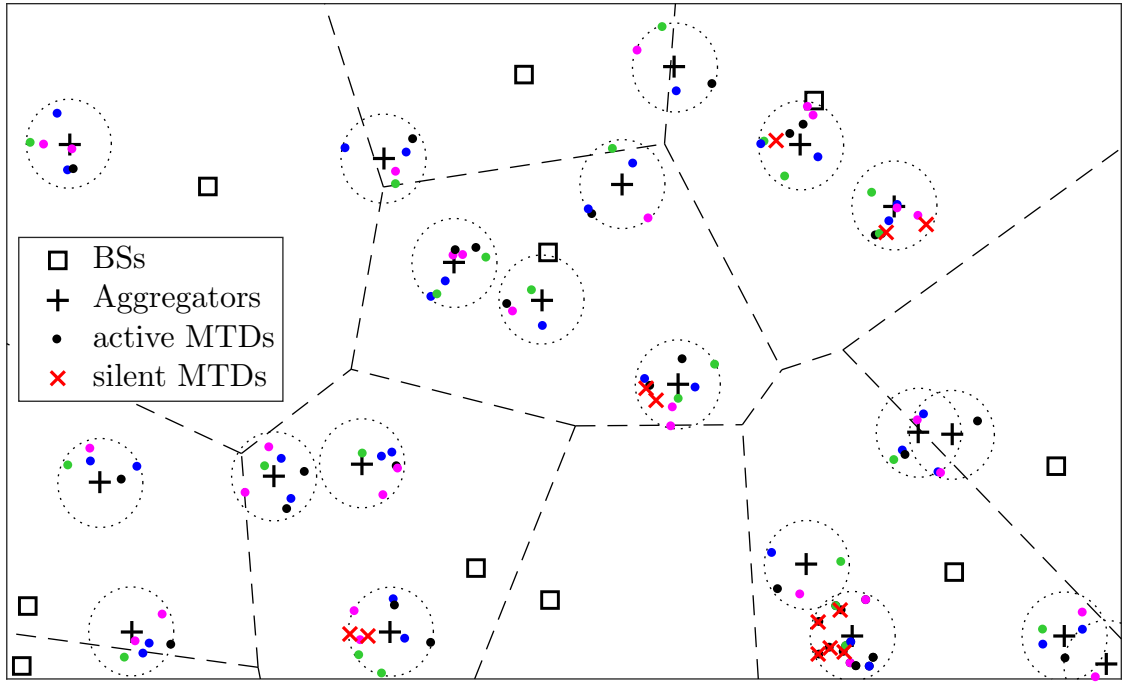}}
	\vspace{-4mm}
	\caption{Snapshot of the system model with $\bar{m}=6$, $L=2$ and $N=4$. MTDs with the same color are using the same channel across the entire network.}		
	\label{Fig1}
	\vspace*{-4mm}
\end{figure}

We assume the quasi-static fading channel model, where the channel power coefficients, $q\in\{h,g\}$, are exponentially distributed with unit mean, e.g., Rayleigh fading. $q=h$ denotes the fading experienced by the signals coming from within the same serving zone, while $q=g$ denotes the fading experienced by the inter-cluster interfering signals. The instantaneous received power at the receiver side is thus given by $p_tqr^{-\alpha}$, where $p_t$ is the transmit power from the transmitter, $r$ is the distance between the receiver and transmitter, and $\alpha$ represents the path-loss exponent. Full channel state information (CSI) is assumed at receiver side as in \cite{Guo.2017,Ding.2014,Zhang.2016} in order to obtain benchmark results, and all MTDs are assumed to use statistical full inversion power control \cite{Gharbieh.2017}, which guarantees a uniform user experience while saving valuable energy, with receiver sensitivity $\rho$. Since here we analyze an interference-limited scenario, $\rho$ does not impact the performance of the network. The aggregators implement the resource scheduling according to one of the schemes presented in the following sections, and the MTDs being considered are those with access granted to the aggregators, since the random access in the network is assumed to be performed\footnote{There are some solutions using NOMA for the random access stage as well, e.g., \cite{Shirvanimoghaddam.2016,Shirvanimoghaddam.2017}. In fact, the work in \cite{Shirvanimoghaddam.2017} proposes a NOMA scheme where the devices transmit their messages over randomly selected channels, while the random access and data transmissions phases are combined. Readers will realize that our resource scheduling schemes can be easily incorporated to that strategy to improve the overall system performance; however, the details of such implementation are out of the scope of this work.}, as in \cite{Guo.2017,Malak.2016,Azari.2016}.
\section{RRS for the Hybrid Access}\label{RRS}
\begin{figure*}[!t]	
	\small
	\begin{equation}\label{Ueq}
	\Pr(U\!=\!u)\!=\!\left\{ \begin{array}{ll}
	Q(N,\bar{m})\big(1-\frac{\bar{m}}{N}\big)+\frac{\exp(-\bar{m})\bar{m}^N}{N!}, &u\!=\!0;\\
	\!\!\sum\limits_{t\in\{\!-1,0,1\}}\!\!\!2^{1\!-\!|t|}\!\bigg(Q\big(N(u\!+\!t),\bar{m}\big)\Big(t\!+\!(\!-\!1)^t\big(\frac{\bar{m}}{N}\!-\!u\big)\Big)\!+\!(\!-\!1)^{1\!-\!|t|}\frac{\exp(\!-\!\bar{m})\bar{m}^{N(u\!+\!t)}}{N(N(u\!+\!t)\!-\!1)!}\bigg)\!&u\!=\!1,...,L\!-\!1\!\\
	1+\!\!\!\sum\limits_{t\in\{0,1\}}\!\!\bigg(Q\big(N(L\!-\!t),\bar{m}\big)\Big((\!-\!1)^t\big(\frac{\bar{m}}{N}\!-\!L\big)\!-\!t\Big)\!+\!(\!-\!1)^{t\!-\!1}\frac{\exp(\!-\!\bar{m})\bar{m}^{N(L\!-\!t)}}{N\big(N(L\!-\!t)\!-\!1\big)!}\bigg),&\! u\!=\!L\ \\ 
	0,& \mathrm{otherwise}
	\end{array},
	\right.
	\end{equation}
	\hrule
\end{figure*}
Herein we explore the RRS scheduling scheme, for which the CSI is only required at the aggregators when decoding the information and not for resource scheduling. Due to its simplicity, this scheme is used as benchmark when compared to a more evolved strategy discussed in Section~\ref{CRS}, where the benefits of the CSI acquisition are exploited also for resource allocation. In Subsection~\ref{RRSa} we attain the success probability for each MTD in a given channel, which depends on the Laplace transform of the inter-cluster interference. An accurate approximation of the latter is given, which is fundamental to efficiently evaluate the success probability. In Subsection~\ref{RRSb} we characterize the overall system performance.

Under the RRS scheme, $N$ out of the $K$ instantaneous MTDs requiring transmissions are independently and randomly, chosen and matched, one-to-one, with the channels in $\mathcal{N}$. If $K\!\le\!N$, all MTDs get channel resources, and even $N\!-\!K$ channels will be unused. Otherwise, if $K\!>\!N$,  the channel allocation is executed again by allowing the remaining MTDs to share channels with the already served MTDs. This process is executed repeatedly until all the MTDs are allocated or the maximum number of MTDs per channel, $L$, is reached for all the channels. 
\begin{lemma}\label{lemma1}
	The Probability Mass Function (PMF) of the number of MTDs allocated to the same channel is given in \eqref{Ueq} at the top of the page.
\end{lemma}
\begin{proof}
	See Appendix~\ref{App_A}.\phantom\qedhere
\end{proof}
The point process of the active MTDs on certain channel $n\in\mathcal{N}$ is obviously a subset of $\Phi$, defined in \eqref{phi1}, and can be defined as
\begin{align}\label{phin}
\Phi_n\delequal\bigcup\limits_{\mathbf{w}\in\Phi_a}\mathbf{w}+\mathcal{B}^{\mathbf{w}}_n,
\end{align}
where $\mathcal{B}^{\mathbf{w}}_n\in\mathcal{B}^{\mathbf{w}}$  denotes the offspring point process with instantaneous number of points around  the cluster center $\mathbf{w}\in\Phi_a$ obeying \eqref{Ueq}. Thus, the generating function of the number of active MTDs on certain channel in one cluster is 
\begin{align}
G_0(z)=\sum_{u=0}^{L}c_uz^u,\label{M}
\end{align}
\noindent where $c_u=\Pr(U=u)$, while $\bar{c}=\sum_{u=1}^{L}uc_u$ is the mean number of active MTDs in each channel of certain cluster.
\subsection{MTD success probability for $L=2$}\label{RRSa}
The intra-cluster interference coming from the MTDs sharing the same channel is faced with SIC\footnote{SIC, as in \cite{Shirvanimoghaddam.2016,Saito.2013_2,Ding.2014,Yang.2016,Abbas.2017,Ding.2017,Zhang.2016,Zhang.2017,Sun.2016,Shirvanimoghaddam.2017_2}, is a common assumption when dealing with NOMA.}. Thus, the SIR\footnote{Other than the real SIR at the receiving antenna of an aggregator, we are more interested in the SIR after SIC that can be used to calculate the success probability. Notice that in the $L=2$ case, the first decoded MTD does not need to perform interference cancellation and directly treats the signal from the second MTD as interference, while the second MTD has to decode first the signal from the other MTD to remove it, which is performed with an efficiency of $100\times(1-\mu)\%$.}, $\mathrm{SIR}^r_{j,u}$, of the $j$th MTD being decoded on a typical channel of certain cluster, given the number of MTDs $u$ sharing the same channel and the RRS scheme is
\begin{align}
\mathrm{SIR}^r_{1,1}&=\frac{h}{I_r},\label{SIR11}\\
\mathrm{SIR}^r_{1,2}&=\frac{\max(h',h'')}{I_r+\min(h',h'')},\label{SIR12}\\
\mathrm{SIR}^r_{2,2}&=\frac{\min(h',h'')}{I_r+\mu\max(h',h'')},\label{SIR22}
\end{align}
where $I_r=\sum_{x\in\Phi_n'}gr_a^{\alpha}x^{-\alpha}$ is the inter-cluster interference for the RRS scheme with $x\!\in\!\Phi_n'\!\subset\!\Phi_n$ denoting both the location and the interfering MTD which occupies certain channel $n$ in other clusters, and $r_a$ is the distance between the MTD and its serving aggregator. Also, $\mu\in[0,1]$ is used to model the impact caused by imperfect SIC \cite{Sun.2016}, while $h'$ and $h''$ are the channel power coefficients of both MTDs sharing the channel when $u=2$. 
Notice that we cannot weight the power of the coexistent nodes on the same channel since CSI information is not exploited for resource scheduling when using the RRS scheme. Also, $\lim\limits_{I_r\rightarrow 0}\mathrm{SIR}_{1,2}^r$ is unbounded, but $\lim\limits_{I_r\rightarrow 0}\mathrm{SIR}_{2,2}^r\le \frac{1}{\mu}$ since $\min(h',h'')\le\max(h',h'')$, thus the performance of the second MTDs being decoded on the channel is strongly limited by the SIC imperfection parameter.

Assuming a fixed rate coding scheme where the receiver decodes successfully whenever $\Pr\left(\mathrm{SIR}\ge\theta \right)$, where $\theta$ is the SIR threshold, e.g., information rate of $\log_2(1+\theta)$ (bits/symbol), we state the following theorem.
\begin{theorem}\label{the1}
	The RRS success probability, $p^r_{j,u}$, of the $j$th MTD sharing a typical channel conditioned on $u$ MTDs, is given by
	\begin{align}
	p^r_{1,1}&=\mathcal{L}_{I_r}(\theta),\label{p11}\\
	p^r_{1,2}&=\!\left\{ \begin{array}{ll}
	\frac{2}{1+\theta}\mathcal{L}_{I_r}(\theta)-\frac{1-\theta}{1+\theta}\mathcal{L}_{I_r}\Big(\frac{2\theta}{1-\theta}\Big),& \mathrm{if}\ 0\le\theta<1\\
	\frac{2}{1+\theta}\mathcal{L}_{I_r}(\theta),& \mathrm{if}\ \theta\ge 1
	\end{array}
	\right.,\label{p12}\\
	p^r_{2,2}&=\!\left\{ \begin{array}{ll}
	\frac{1-\theta\mu}{1+\theta\mu}\mathcal{L}_{I_r}\Big(\frac{2\theta}{1-\theta\mu}\Big),& \mathrm{if}\ 0\le\theta\mu<1\\
	0,& \mathrm{if}\ \theta\mu\ge 1
	\end{array}
	\right.,\label{p22}	
	\end{align}
	\noindent
	where
		\begin{align}
	\mathcal{L}_{I_r}(s)&\!=\!\exp\!\left(\!2\pi\lambda_a\!\int\limits_{0}^{\infty}\!r_{\mathbf{w}}\Big(\sum\limits_{u=0}^{L}\!c_u\Upsilon(r_{\mathbf{w}},s)^u\!-\!1\Big)\mathrm{d}r_{\mathbf{w}}\!\right)\!,\label{LI}
	\end{align}	
	\noindent is the Laplace transform of RV ${I_r}$ and
	\begin{align}
	\Upsilon(\!r_{\mathbf{w}},s\!)&\!=\!\frac{1}{\pi R_a^2}\!\!\int\limits_{0}^{R_a}\!\!\int\limits_{0}^{2\pi}\!\!\frac{r_a\mathrm{d}\omega\mathrm{d}r_a}{\!1\!+\!sr_a^{\alpha}\!\big(r_{\mathbf{w}}^2\!+\!r_a^2\!+\!2r_{\mathbf{w}}r_a\cos(\omega)\!\big)^{\!-\frac{\alpha}{2}}}.\label{ups}
	\end{align}
\end{theorem}
\begin{proof}
	See Appendix~\ref{App_B}.\phantom\qedhere
\end{proof}
\begin{remark}
	Although for some pairs $(\theta,\mu)$ it is possible to achieve a similar reliability of both MTDs sharing the same channel, e.g., $p^r_{1,2}\approx p^r_{2,2}$, this is not attained in general, and it could be a main drawback when using the RRS scheme and certain homogeneity in the Quality of Service (QoS) is expected. Another interesting fact is that $p^r_{1,1}=p^r_{1,2}=\mathcal{L}_{I_r}(\theta)$ for $\theta=1$, while $p^r_{2,2}$ would be smaller even when $\mu\rightarrow 0$, e.g., $\mathcal{L}_{I_r}(2\theta)$.
	Additionally, notice that $\mu<1/\theta$ is required such that the second user being decoded on certain channel has chances to succeed. Thus, the greater the SIR threshold, $\theta$, the greater the impact of imperfect SIC, $\mu$.
\end{remark}
\begin{theorem}\label{lemma2}
	Expression $\eqref{LI}$ for $L\ge 2$ is upper (almost surely) and lower bounded by 
	\begin{align}
	\mathcal{L}_{I_r}(s)&\stackrel{a.s}{\le}\mathcal{L}_{I_r}^{\mathrm{up}}(s)=\exp\Big(-\chi s^{\frac{2}{\alpha}}\sum\limits_{u=1}^{L}c_uu^{\frac{2}{\alpha}}\Big),\label{up}\\
	\mathcal{L}_{I_r}(s)&\ge\mathcal{L}_{I_r}^{\mathrm{lo}}(s)=\exp\big(-\chi\bar{c} s^{\frac{2}{\alpha}}\big),\label{lo}
	\end{align}
	where $\chi=\frac{1}{2}\lambda_a\pi R_a^2\Gamma\big(1+\tfrac{2}{\alpha}\big)\Gamma\big(1-\tfrac{2}{\alpha}\big)$, while 
	\begin{align}
	\mathcal{L}_{I_r}(s)&\approx\beta_0\mathcal{L}_{I_r}^{\mathrm{up}}(s)+\beta_1\mathcal{L}_{I_r}^{\mathrm{lo}}(s),\label{ap}
	\end{align}
	provides an approximation with $\beta_0,\beta_1\in[0,1]$ and $\beta_0+\beta_1\!=\!1$. 
\end{theorem}
 \begin{proof}
 	See Appendix~\ref{New} for derivation of \eqref{up} and \eqref{lo}, while \eqref{ap} is a weighted average of both bounds, thus attaining an approximation that will be more accurate than at least one of the bounds. The weighted average becomes relevant if we know a priori which of the bounds in \eqref{up} and \eqref{lo} is more accurate for the setup being analyzed and we give it a greater weight, e.g., $\beta_0\lessgtr\beta_1$, otherwise trivial choice $\beta_0=\beta_1=0.5$ is advisable, as we adopt here.
 \end{proof}
 \begin{figure}[t!]
 	\centering
 	\subfigure{\label{Fig2a}\includegraphics[width=0.46\textwidth]{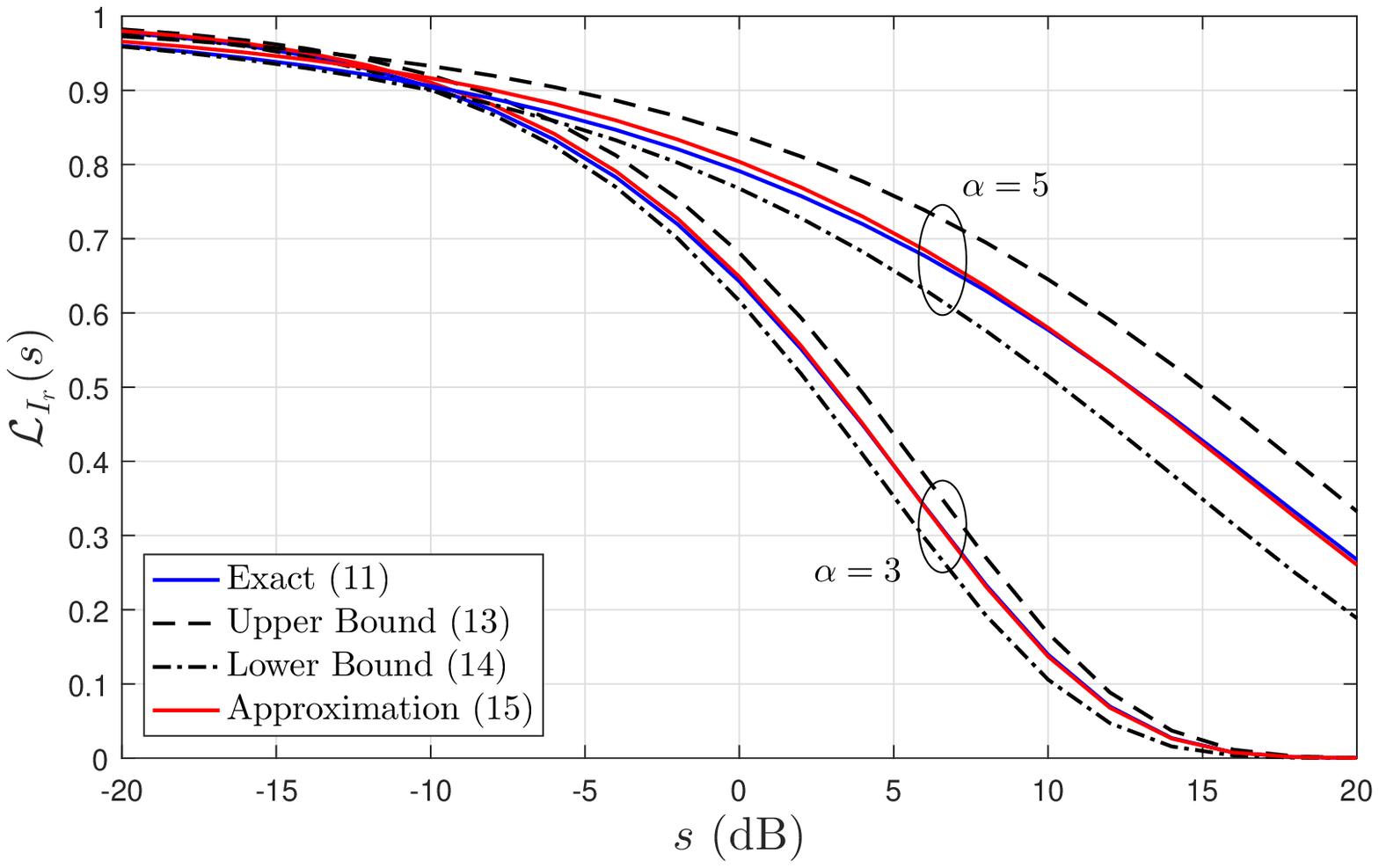}}\ \ \ 
 	\subfigure{\label{Fig2b}\includegraphics[width=0.46\textwidth]{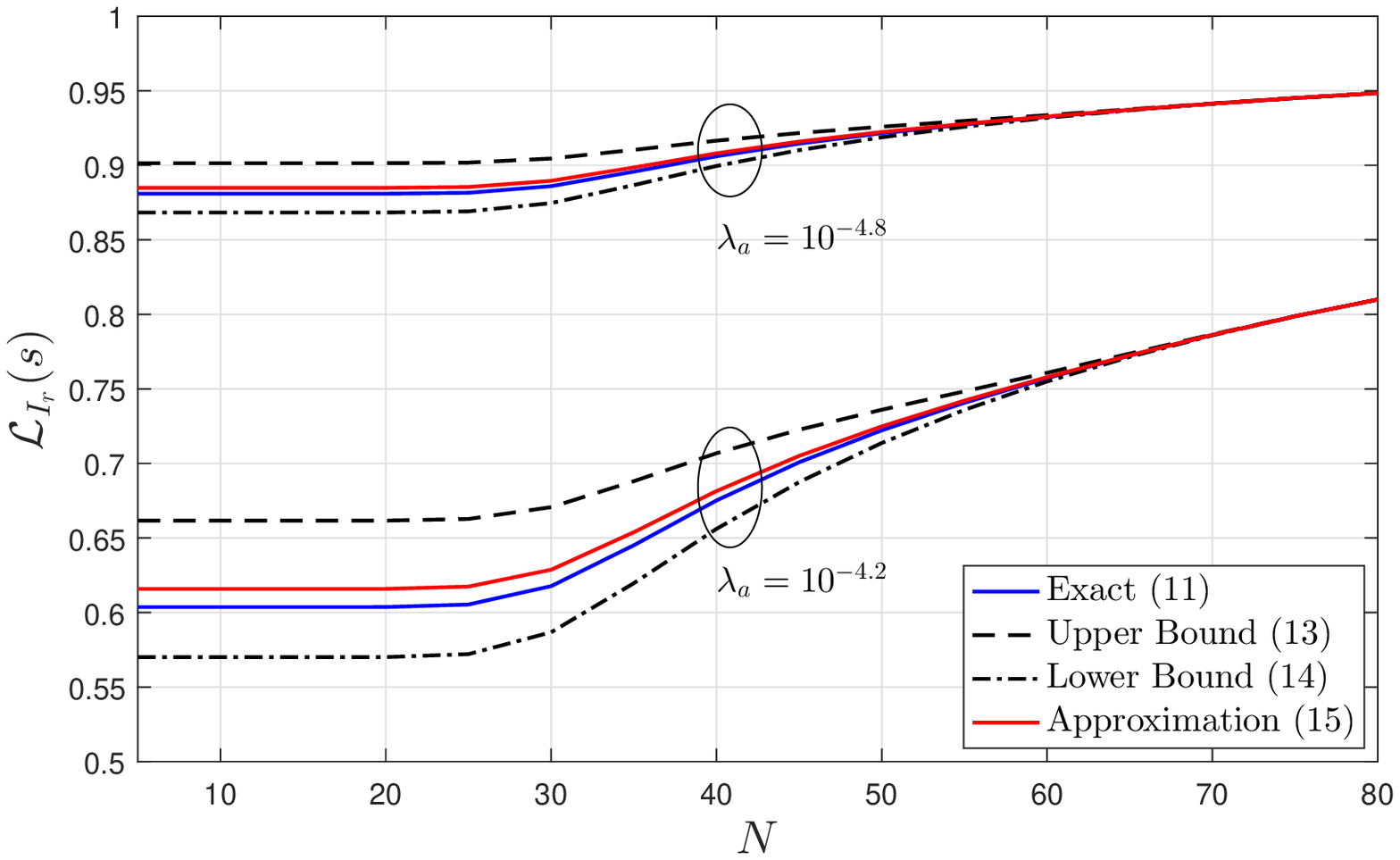}}
 	\vspace{-4mm}
 	\caption{Comparing the exact $\mathcal{L}_{I_r}(s)$ \eqref{LI} as a function of (a) $s$ in dB, for $c_0=c_1=0,\ c_2=1$, $\lambda_a=10^{-4.4}$ (left), and (b) $N$ for $\bar{m}=60$, $\alpha=3.6$, $\lambda_a\in\{10^{-4.2},10^{-4.8}\}$ and $s=0$ dB (right), with the bounds and approximation given in Theorem~\ref{lemma2}. We also set $R_a=40$m.}	
 	\label{Fig2}	 	
 	\vspace*{-4mm}
 \end{figure}
 Fig.~\ref{Fig2} shows the bounds and approximation attained in Theorem~\ref{lemma2} while comparing them with the exact value given in \eqref{LI}.  Specifically, Fig.~\ref{Fig2a} shows the performance for $\alpha\in\{3,5\}$, while we set $c_2=1$, which conduces to scenarios where the bounds in \eqref{up} and \eqref{lo} are the least tight since $T_1$ in \eqref{LirAp} has no impact on \eqref{LI}, e.g., $T_1=1$. As noticed, the greater the $\alpha$ the lesser the tightness of the bounds, however the approximation performs very well whatever the setup, which is also appreciated in Fig.~\ref{Fig2b}.  Notice that for relatively small $N$, the approximation in \eqref{ap} is accurate even when $\lambda_a$ increases and the bounds are not tight. When $N$ increases, $c_0$ and $c_1$ increase, which increases the contribution of $T_1$ in \eqref{LirAp}, thus increasing the tightness of the bounds, and the accuracy of the approximation. 
\subsection{Overall Performance for $L=2$}\label{RRSb}
The following two lemmas characterize the system performance in terms of average over all MTD success probability, and average number of simultaneously served MTDs since one of the main benefits of NOMA techniques is the possibility of offering service to a great number of devices simultaneously.
\begin{lemma}
	Conditioned on being using certain channel, the average over all MTD success probability is
	\begin{equation}\label{prrs}
	p^r_{\mathrm{succ}}=\frac{c_1}{1-c_0}p^r_{1,1}+\frac{c_2}{2(1-c_0)}(p^r_{1,2}+p^r_{2,2}).
	\end{equation}
\end{lemma}
\begin{proof}
	Knowing that $c_u=\Pr(U=u)$ it is straightforward to attain \eqref{prrs}.
\end{proof}
\begin{figure*}[!t]	
	\footnotesize
	\begin{align}\label{Kr}
	\bar{K}_r&\!\stackrel{(a)}{=}\!p_{1,1}^r\sum_{k=0}^{N}k\Pr(K\!=\!k)\!+\!p^r_{1,1}\sum_{k\!=\!N\!+\!1}^{2N-1}(2N\!-\!k)\Pr(K\!=\!k)\!+\!(p^r_{1,2}\!+\!p^r_{2,2})\!\sum_{k\!=\!N+1}^{2N\!-\!1}\!(k\!-\!N)\!\Pr(K\!=\!k)\!+\!N(p^r_{1,2}\!+\!p^r_{2,2})\!\sum_{k\!=\!2N}^{\infty}\!\Pr(K\!=\!k)\nonumber\\
	&\!\stackrel{(b)}{=}\!p^r_{1,1}(A_1+A_2)+(p^r_{1,2}+p^r_{2,2})(A_3-A_2)
	\end{align}
	\hrule
\end{figure*}
\begin{lemma}
	The average number of simultaneously served MTDs is given in \eqref{Kr} at the top of the page, where
	\begin{align}    
    A_1&=\bar{m}Q(N+1,\bar{m})-\frac{\exp(-\bar{m})\bar{m}^{N+1}}{N!},\label{A1}\\    
    A_2&=\exp(-\bar{m})\Big(\frac{\bar{m}^{2N}}{(2N-1)!}-\frac{\bar{m}^{N+1}}{N!}\Big)-(\bar{m}-2N)\Big(Q(2N,\bar{m})-Q(N+1,\bar{m})\Big),\label{A2}\\    
	A_3&=N\big(Q(N+1,\bar{m})+1\big)
	\end{align}
\end{lemma}
\begin{proof}
	In \eqref{Kr}, $(a)$ comes from averaging the number of simultaneously active and successful MTDs, while $(b)$ comes from regrouping terms and using the expressions of the CDF and expected $K$ on one interval (see Appendix~\ref{App_A}) along with some algebraic transformations.
\end{proof}
\begin{remark}\label{remark3}
	When $\bar{m}$ grows above $2N$, the last term in equality $(a)$ of \eqref{Kr} becomes the largest contributor for $\bar{K}_r$, with $\bar{K}_r\rightarrow N(p^r_{1,2}+p^r_{2,2})=2Np^r_{\mathrm{succ}}$ since $\sum_{k=2N}^{\infty}\Pr(K=k)\approx 1$ and $c_2\approx 1$ ($c_0\approx c_1\approx 0$) in \eqref{prrs}. Thus, when comparing with an $L=1$ setup under the same circumstances, e.g., with average number of simultaneously served MTDs being $Np_{\mathrm{succ}}^{\mathrm{OMA}}$, the condition required for the hybrid access scheme\footnote{In general, the required condition can be extended for any $L$, thus $p_{\mathrm{succ}}^{\mathrm{HYB}}>\tfrac{1}{L}p_{\mathrm{succ}}^{\mathrm{OMA}}$.} to overcome the OMA system ($L=1$) is that $p^r_{\mathrm{succ}}>\tfrac{1}{2}p_{\mathrm{succ}}^{\mathrm{OMA}}$. 
    Notice that the hybrid scheme with RRS could lead to a less reliable system with greater chances of outages per MTD, e.g., due to a larger inter-cluster interference  and additional intra-cluster interference. However, the number of simultaneous active MTDs could be improved for high access demand scenarios, e.g., $\bar{m}\sim 2N$, as long as the success probability does not decrease so much.	
\end{remark}
\section{CRS for the Hybrid Access}\label{CRS}
Herein we explore the CRS scheduling scheme, which contrary to the RRS scheme, strongly relies on the CSI for resource scheduling. Thus, the CRS scheme is more adjusted to NOMA scenarios, where CSI is a keystone for efficiently decoding multiple user data over the same orthogonal channel with SIC \cite{Ding.2014,Zhang.2016,Imari.2014}. Similar than in the previous section, we attain first in Subsection~\ref{CRSa} the success probability for each MTD in a given channel along with necessary approximations. Later, in Subsection~\ref{CRSb}, we find the power constraints on the MTDs sharing the same channel in order to attain a fair coexistence of our scheme with purely OMA setups, while power control coefficients are found too, so these MTDs can perform with similar reliability. Finally, in Subsection~\ref{CRSb} we characterize the overall system performance.

Under the CRS scheme, the MTD with better fading (equivalently, better SIR) will be preferentially assigned with the available channel resources. An aggregator with $K$ instantaneous  MTDs requiring transmission has the knowledge of their fading gains. Let $\{h_1,...,h_i,...,h_K\}$ denote the decreasing ordered channel gains, where $h_{i-1}>h_i$. If $K\le N$ all the MTDs will be chosen, but if $K>N$ the aggregator will pick the $N$ MTDs with better channel gains, i.e., $h_1,...,h_N$, and then will assign the channel set $\mathcal{N}$ to them \cite{Guo.2017}. As a continuation, the remaining MTDs can  be still allocated sharing those same resources, i.e., users $N+1$,...,$K$ go to the second round for allocation. This process is executed repeatedly until all the MTDs are allocated or the maximum number of MTDs per channel, $L$, is reached. 
Both Lemma~\ref{lemma1} and the point process of the active MTDs on certain channel $n\in\mathcal{N}$ characterized through \eqref{phin} and \eqref{M} hold here.
\subsection{MTD success probability for $L=2$}\label{CRSa}
Under the CRS scheme and using SIC to face the intra-cluster interference, the SIR, $\mathrm{SIR}^{c\ (i)}_{j,u}$, of the $j$th MTD being decoded on a typical channel, given the first MTD allocated there has the $i$th larger channel coefficient, $h_i$, and there are $u$ MTDs sharing that same channel, is given by
\begin{align}\label{SIRcrs}
\mathrm{SIR}^{c\ (i)}_{1,1}&=\frac{h_i}{I_c},\\
\mathrm{SIR}^{c\ (i)}_{1,2}&=\frac{a_ih_i}{I_c+b_ih_{i+N}},\label{sir12}\\
\mathrm{SIR}^{c\ (i)}_{2,2}&=\frac{b_ih_{i+N}}{I_c+\mu a_ih_i}.\label{sir22}
\end{align}
Notice that differently from the RRS scheme, here we can weight the power of coexistent nodes on the same channel through $a_i$ and $b_i$ since CSI is used for resource allocation. Of course, some kind of feedback from the aggregators would be required.
Once again, the performance of the first MTD being dedoced in the channel is somewhat unbounded, e.g., $\lim\limits_{I_c\rightarrow 0}\mathrm{SIR}_{1,2}^{c\ (i)}$ is unbounded, while the performance of the second one is not,  but this time we can relax that situation by choosing $a_i<b_i$ since $\lim\limits_{I_c\rightarrow 0}\mathrm{SIR}_{2,2}^{c\ (i)}<\frac{b_i}{\mu a_i}$.
However, weighting the transmit power of the MTDs sharing the same channel conduces to a marked point process with marks $m_x=1$ when the MTD on $x$ is alone in the channel, which happens with probability $c_1$, and $m_x\in\{a_i,b_i\}$ for those MTDs sharing the same channel, which happens with probability $c_2$. Thus, the inter-cluster interference for the CRS scheme is $I_c=\sum_{x\in\Phi_n'}gm_xr_a^{\alpha}x^{-\alpha}$.
By letting $a_i+b_i=\delta$ be a fixed value we impose some kind of total transmission power constraint \cite{Yang.2016}. This is crucial for NOMA scenarios, and here it is particular important in order to control the interference from the inter/intra-cluster MTDs sharing the same channel.  
Now, assuming that the receiver can decode successfully (SIR exceeds a threshold $\theta$), we state the following theorem.
\begin{theorem}\label{the3}
	The CRS success probability, $p^{c\ (i)}_{j,u}$, of the $j$th MTD being decoded on a typical channel, given the first MTD allocated there has the $i$th larger channel coefficient, $h_i$, and there are $u$ MTDs sharing that same channel, is approximately given by
	\begin{align}
	p^{c\ (i)}_{j,u}&\!\approx\!\frac{1}{2}\!-\!\frac{1}{\pi}\!\int\limits_{0}^{\infty}\!\frac{1}{\varphi}\mathrm{Im}\Big\{\mathcal{L}_{I_c}(\!-\mathbf{i}\varphi)\exp\!\big(\!\!-\mathbf{i}\varphi B_{j,u}^{(i,K)}\big)\!\Big\}\mathrm{d}\varphi,\label{pc}
	\end{align}
	\noindent where
	\begin{align}
	\mathcal{L}_{I_c}(s)&=\!\exp\!\!\left(\!\!2\pi\lambda_a\!\!\int\limits_{0}^{\delta}\!\!\int\limits_{0}^{\infty}\!\!r_{\mathbf{w}}\big(c_0\!+\!c_1\Upsilon(r_{\mathbf{w}},s)\!+\!c_2\Upsilon(r_{\mathbf{w}},a_is)\Upsilon\big(r_{\mathbf{w}},(\delta\!-\!a_i)s\big)\!-\!1\big)f_{a_i}(a_i)\mathrm{d}r_{\mathbf{w}}\mathrm{d}a_i\!\!\right)\!\!,\label{LI2}\\	
	&\approx\sum_{t\in\{1,2\}}\beta_{t-1}\exp\left(-\chi\big(c_1+c_2t^{\frac{\alpha-2}{\alpha}}\delta^{\frac{2}{\alpha}}\big)s^{\frac{2}{\alpha}}\right),\label{lm7}
	\end{align}
	\noindent is the Laplace transform of RV ${I_c}$ for $L=2$, $\Upsilon(r_{\mathbf{w}},s)$ is defined in \eqref{ups}, and $\beta_0,\beta_1\in[0,1]$, $\beta_0+\beta_1=1$, such that \eqref{lm7} with $\beta_0=1$ and $\beta_1=1$ provide, \textit{almost surely}, upper and lower bounds for \eqref{LI2}, respectively. Finally,
	\begin{align}
	B_{1,1}^{(i,K)}&=\frac{\psi(K+1)-\psi(i)}{\theta},\label{B11}\\
	B_{1,2}^{(i,K)}&=\Big(\frac{a_i}{\theta}\!-\!b_i\Big)\psi(K+1)\!+\!b_i\psi(i+N)\!-\!\frac{a_i}{\theta}\psi(i),\label{B12}\\
	B_{2,2}^{(i,K)}&=\Big(\frac{b_i}{\theta}\!-\!\mu a_i\Big)\psi(K\!+\!1)\!+\!\mu a_i\psi(i)\!-\!\frac{b_i}{\theta}\psi(i\!+\!N).\label{B22}
	\end{align}
\end{theorem}
\begin{proof}
	See Appendix~\ref{App_D}.\phantom\qedhere
\end{proof}
Notice that even in the case when all MTDs operating on the same channel in the network are using the same $a_i$ and $b_i$ coefficients, e.g., $a_i$ and $b_i$ with deterministic values, evaluating \eqref{pc} using \eqref{LI2} is not an easy task. A closer look at that expression makes us suspect on its efficiency. This is because it requires evaluating two inner triple integrals.
In fact, several numerical tests we run corroborated that evaluating \eqref{pc}, using \eqref{LI2} with fixed $a_i,b_i$ values, is highly inefficient and computationally too heavy. Thus, the approximation given in \eqref{lm7}  becomes necessary, which holds also under the premise that the pair $(a_i,b_i)$ has not to be simultaneously the same for all MTDs operating on the same channel. This would be required for instance in order to optimize in some way the system performance. Additionally, notice that $B_{1,2}^{(i,K)}$ and $B_{2,2}^{(i,K)}$ defined in \eqref{B12} and \eqref{B22}, respectively, are functions of the power control coefficients $a_i$ and $b_i$ of the typical links, which are assumed fixed for the entire period of evaluation.

\begin{figure}[t!]
	\centering
	\subfigure{\label{Fig3a}\includegraphics[width=0.46\textwidth]{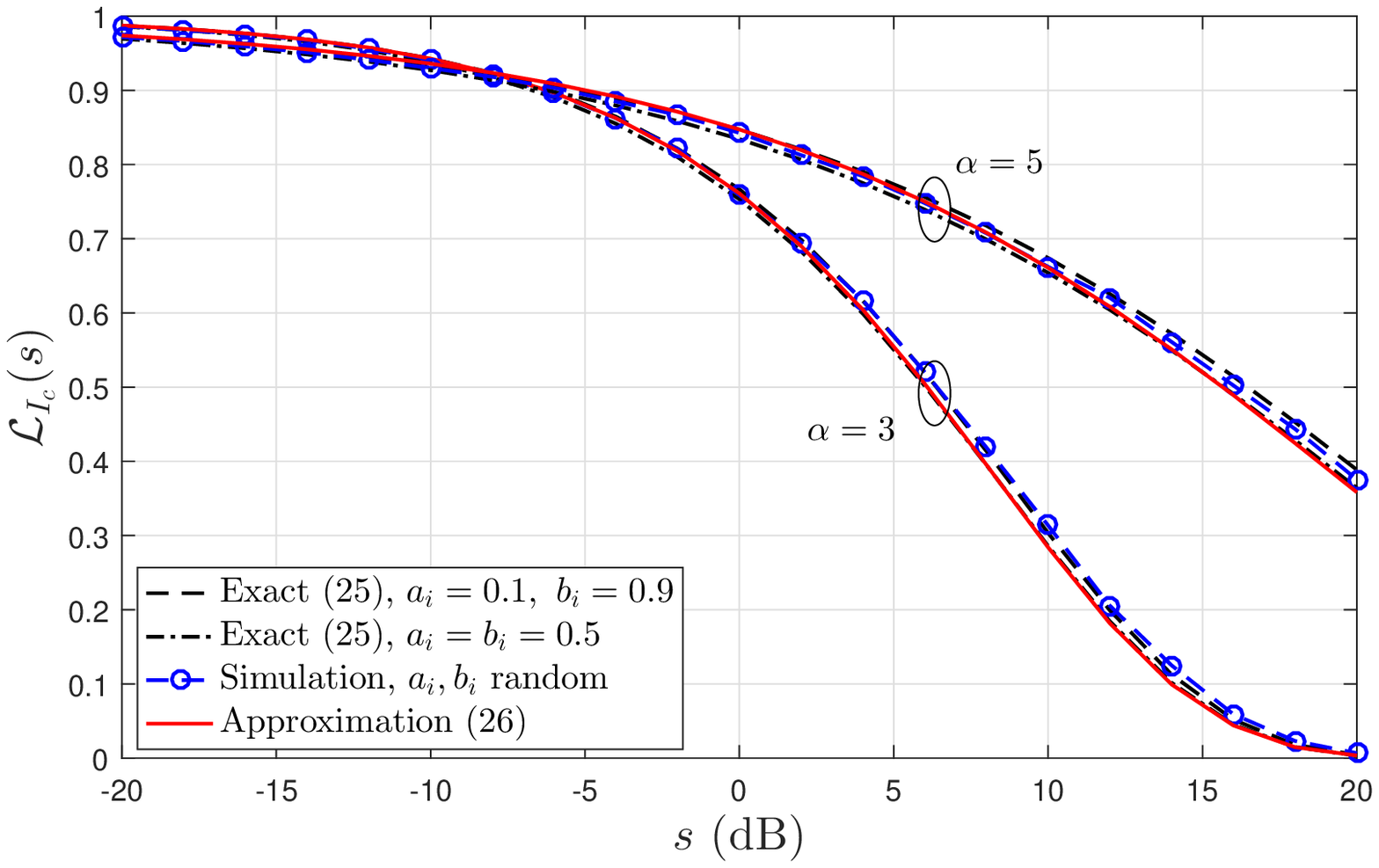}}\ \ \ 
	\subfigure{\label{Fig3b}\includegraphics[width=0.46\textwidth]{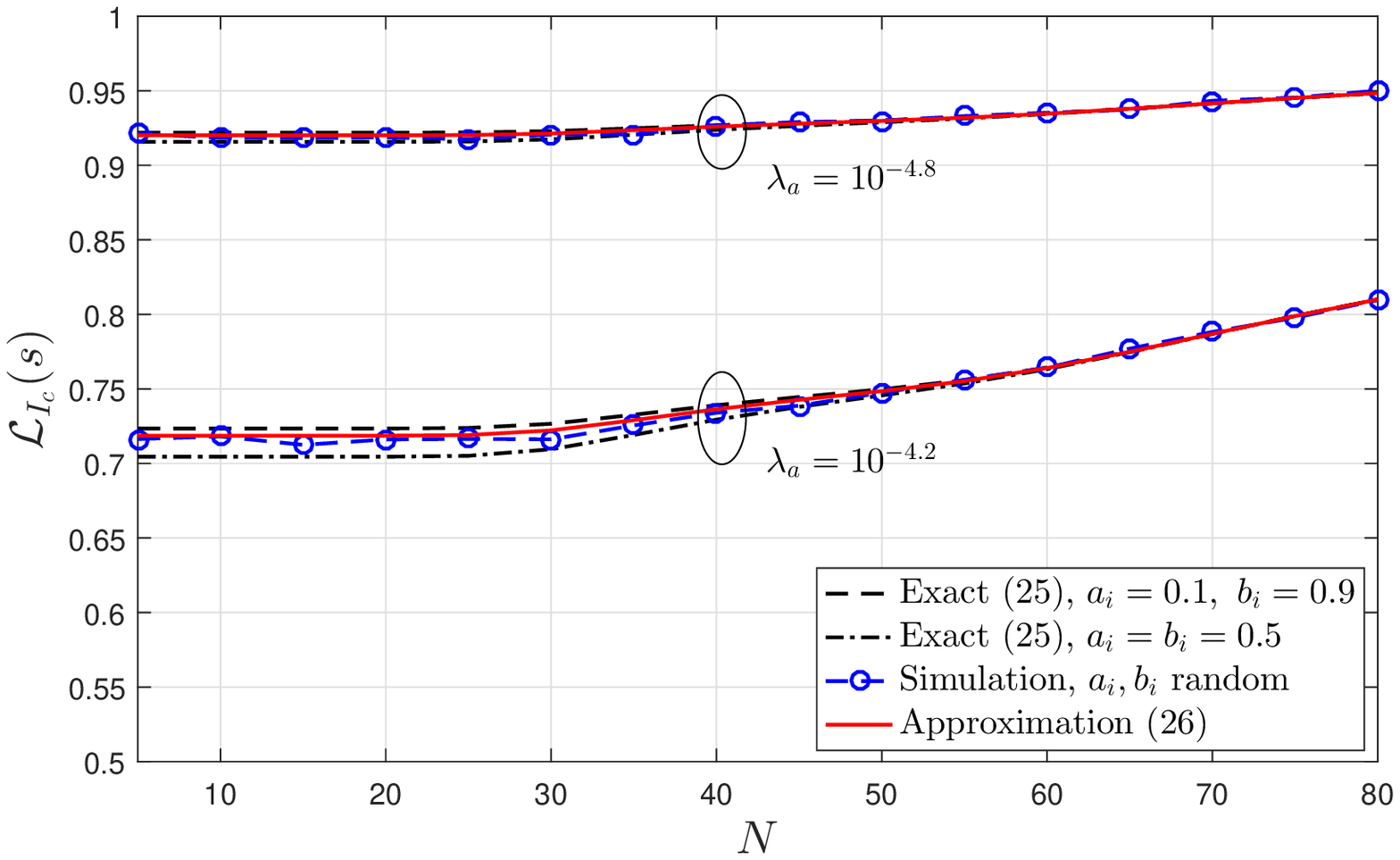}}
	\vspace{-4mm}
	\caption{Comparing the exact $\mathcal{L}_{I_c}(s)$ \eqref{LI2} as a function of (a) $s$ in dB, for $c_0=c_1=0,\ c_2=1$, $\lambda_a=10^{-4.4}$ (left), and (b) $N$ for $\bar{m}=60$, $\alpha=3.6$, $\lambda_a\in\{10^{-4.2},10^{-4.8}\}$ and $s=0$ dB (right), with the approximation given in \eqref{lm7} and $\delta=1$. We also set $R_a=40$m, and simulations are with $a_i,b_i$ uniformly distributed in $[0,1]$.}
	\label{Fig3}	 	
	\vspace*{-4mm}
\end{figure}

Fig.~\ref{Fig3} shows the accuracy of \eqref{lm7} by plotting the exact values using \eqref{LI2} for fixed $a_i,b_i$ and $\delta=1$ and comparing also with simulations by choosing $a_i$, $b_i$ uniformly random. The values for $a_i$, $b_i$ are interchangeable since the perceived interference is the same. The remaining values for the system parameters are the same than the previously used when discussing Fig.~\ref{Fig2}. Both, Fig.~\ref{Fig3a} and Fig.~\ref{Fig3b} showing $\mathcal{L}_{I_c}$ as a function of $s$ and $N$, respectively, corroborate the idea behind  \eqref{lm7}. This is, the system performance depends heavily on $\delta$ rather than the individuals $a_i$, $b_i$, or their distribution. Notice that the exact values for two completely different pairs $(a_i,b_i)$ are very similar, and even when $a_i$, $b_i$ were chosen randomly, the performance is kept alike.

\begin{remark}\label{rem}
	We know that the interference under the CRS scheme with $a_i=b_i=1$ is characterized in the same way as in the RRS case, $\mathcal{L}_{I_c}=\mathcal{L}_{I_r}$, since no weights to the transmit power, e.g., no marks, are assigned; and notice that \eqref{lm7} captures well this phenomena since with $\delta=2$ we attain \eqref{ap}. This means that whatever the values of $a_i$, $b_i$, as long as $\delta=2$ the interference is distributed approximately equal as for the RRS scenario.
\end{remark}

\begin{corollary}\label{p2}
	An alternative, and easy of evaluating, expression for the CRS success probability, $p_{j,u}^{c\ (i)}$, for a given $\delta$, is
	\begin{align}
	p_{j,u}^{c\ (i)}\approx\frac{1}{2}-\sum_{t\in\{1,2\}}\frac{\beta_{t-1}}{\pi}\int\limits_{0}^{\infty}\frac{\exp(-\varsigma_{t}\varphi^{\frac{2}{\alpha}})\sin(\varrho_{t}\varphi^{\frac{2}{\alpha}}-\varphi B_{j,u}^{i,K})}{\varphi}\mathrm{d}\varphi,\label{ap3}
	\end{align}	
	where $\varsigma_r=\nu_r\cos\big(\tfrac{\pi}{\alpha}\big)$, $\varrho_r=\nu_r\sin\big(\tfrac{\pi}{\alpha}\big)$ and $\nu_r=\chi(c_1+c_2r^{\frac{\alpha-2}{\alpha}}\delta^{\frac{2}{\alpha}})$.
\end{corollary}
\begin{proof}
	The idea here is substituting into \eqref{pc} the approximation for $\mathcal{L}_{I_c}(s)$ given in \eqref{lm7}. Using $(-\mathbf{i})^{\frac{2}{\alpha}}=\cos(\tfrac{\pi}{\alpha})-\mathbf{i}\sin(\tfrac{\pi}{\alpha})$ and $\mathrm{Im}\{p\exp(-q\mathbf{i})\}=-p\sin(q)$, along with simple algebraic transformations, we attain \eqref{ap3}. The summation comes from both sum terms in \eqref{lm7}.
\end{proof}

\begin{figure*}[!t]
	\small
	\begin{align}
	p_{1,1}^c&\stackrel{(a)}{=}\frac{1}{1-\Pr(K=0)}\bigg(p_{1,1}^r\Pr(K\le N)+\sum_{k=N+1}^{2N-1}\sum_{i=\mathrm{mod}(k,N)+1}^{N}p_{1,1}^{c\ (i)}\frac{\Pr(K=k)}{N-\mathrm{mod}(k,N)}\bigg)\nonumber\\
	&\stackrel{(b)}{=}\frac{1}{\exp(\bar{m})-1}\Big(p_{1,1}^rQ(N+1,\bar{m})+\sum_{k=N+1}^{2N-1}\sum_{i=k-N+1}^{N}p_{1,1}^{c\ (i)}\frac{\bar{m}^k}{(2N-k)k!}\Big),\label{p11c}\\
	p_{j,2}^c&\stackrel{(a)}{=}\frac{1}{\Pr(K>N)}\bigg(\sum_{k=N+1}^{2N-1}\sum_{i=1}^{\mathrm{mod}(k,N)}p_{j,2}^{c\ (i)}\frac{\Pr(K=k)}{\mathrm{mod}(k,N)}+\frac{1}{N}\sum_{k=2N}^{\infty}\sum_{i=1}^{N}p_{j,2}^{c\ (i)}\Pr(K=k)\bigg)\nonumber\\
	&\stackrel{(b)}{\approx}\frac{\exp(-\bar{m})}{1-Q(N+1,\bar{m})}\bigg(\sum_{k=N+1}^{2N-1}\sum_{i=1}^{k-N}p_{j,2}^{c\ (i)}\frac{\bar{m}^k}{(k-N)k!}+\frac{1}{N}\sum_{k=2N}^{k_{\max}}\sum_{i=1}^{N}p_{j,2}^{c\ (i)}\frac{\bar{m}^k}{k!}\bigg)\label{p12p22c}.
	\end{align}
	\hrule
\end{figure*}	
\begin{lemma}
	The CRS success probability, $p_{j,u}^c$, of the $j$th MTD sharing a typical channel conditioned on $u$ MTDs, is given in \eqref{p11c} and \eqref{p12p22c} at the top of the next page, where $j\in\{1,2\}$.
\end{lemma}
\begin{proof}
	Note that $(a)$ comes directly from 
	the total probability theorem, and also averaging $p_{j,u}^{c\ (i)}$ over all possible values of $i$ given $k$ and $N$. On the other hand, $(b)$ comes from using the expressions for $\Pr(K\le N)$ and $\Pr(K=k)$ easily obtained through the PDF and CDF of $K$. Also, $\mod(k,N)=k-N$ if $N+1 \le k\le 2N-1$, and the approximation is because we substituted the infinite sum  for a finite sum until $k_{\max}$ to reach \eqref{p12p22c}.
\end{proof}
Notice that a proper value for $k_{\max}$ is such that the sum over the remaining $k>k_{\max}$ does not contribute significantly to the success value in \eqref{p12p22c}. Thus, we could choose $k_{\max}$ such $\sum_{k=k_{\max}+1}^{\infty}\Pr(K=k)<\tau\rightarrow Q(k_{\max}+1,\bar{m})>1-\tau$ holds, where $\tau$ is the maximum allowable error we admit when approximating with \eqref{p12p22c}, e.g., $\tau=10^{-5}$ and $\bar{m}=30\rightarrow k_{\max}=56$.
\subsection{Practical issues}\label{CRSb}
	If we look closely at \eqref{sir12} and \eqref{sir22} we can notice there must exist some choice for $a_i$ and $b_i$ in order to attain a similar reliability for both MTDs sharing the orthogonal channel. Thus, we resort to $\mathrm{SIR}_{1,2}^{c\ (i)}=\mathrm{SIR}_{2,2}^{c\ (i)}$ with $a_i+b_i=\delta$. However, the knowledge of the inter-cluster interference is required and it is a major drawback for this \begin{flushright}
		method. Alternatively, we state the following theorem avoiding that problem.
	\end{flushright}
\begin{theorem}\label{the4}
	A proper approximate choice for $a_i$ and $b_i$ in order to attain a similar reliability for both MTDs sharing the channel when $u=2$ is given by
	\begin{align}
	a_i&\!=\!\frac{\delta\big(1+\frac{1}{\theta}\big)\big(\psi(K+1)-\psi(i+N)\big)}{\!\big(\!1\!+\!\mu\!+\!\frac{2}{\theta}\!\big)\psi(\!K\!+\!1\!)\!-\!\big(\!\mu\!+\!\frac{1}{\theta}\!\big)\psi(i)\!-\!\big(\!1\!+\!\frac{1}{\theta}\!\big)\psi(i\!+\!N)\!}\label{ai},\\
	b_i&=\delta-a_i.\label{bi}
	\end{align}
\end{theorem}
\begin{proof}
	Notice that $p_{1,2}^{c\ (i)}$ and $p_{2,2}^{c\ (i)}$ in \eqref{pc} only differ in the terms $B_{1,2}^{(i,K)}$ and $B_{2,2}^{(i,K)}$. Thus, solving $B_{1,2}^{(i,K)}=B_{2,2}^{(i,K)}$ with $a_i+b_i=\delta$ conduces to an approximate choice for these parameters. 
\end{proof}
\begin{remark}
	The values of $a_i$ and $b_i$ aside of the index $i$, strongly rely on the instantaneous number of MTDs contesting for transmission resources, $K$, the number of available channels, $N$, the SIR threshold, $\theta$, and the SIC imperfection parameter, $\mu$. These values for $a_i,\ b_i$ do not guarantee the same instantaneous SIR  for both MTDs sharing the channel, but when averaging over a long period\footnote{Assuming that each of them will be occupying the same $i$th channel and the same order when transmitting in future rounds.} it guarantees a similar reliability for them. Also, all MTDs have the same chances of occupying any of the channels in $\mathcal{N}$, thus for high loaded systems where $\bar{m}>N$, e.g., the probabilities of all the channels being occupied by two MTDs are great, an average similar reliability will be attained. 
\end{remark}

On the other hand, both the RRS scheme, and CRS scheme with relatively large $\delta$, do not favor the coexistence with purely OMA setups. The reason is because these NOMA setups would increase the interference, e.g., up to twice and approximately $\delta$ larger for RRS and CRS schemes, respectively, caused to the OMA setups, and since this is not compensated by multiplexing several users, the performance of the OMA clusters is affected. On the other hand, even when only our hybrid OMA-NOMA scheme is utilized, it is expected that all aggregators will not be under the same average access demand, e.g., same $\bar{m}$.
	Therefore, those with low demand could be operating with OMA almost all the time and the larger interference of previously schemes will be impractical. 
Thus, limiting the interference is crucial and that could be done by properly selecting a relatively small $\delta$.

\begin{theorem}\label{the6}
	The required $\delta$, $\delta^*,$ for a fair coexistence\footnote{We refer to ``fair coexistence'' when for any aggregator, the interference coming from the outside topology (the inter-cluster interference) remains the same regardless of the alternative (OMA protocol or our hybrid approach) the outside clusters utilize.} between OMA and our hybrid setup is approximated by the solution of	
	\begin{equation}
	\xi^{\delta^{\frac{2}{\alpha}}-1}+\xi^{2^{\frac{\alpha-2}{\alpha}}\delta^{\frac{2}{\alpha}}-1}=2,\label{eq}
	\end{equation}
	where $\xi=\exp(-\chi c_2 s^{\frac{2}{\alpha}})$, and it is bounded by
	\begin{equation}
	2^{\frac{2-\alpha}{2}}\le\delta^*\le 1.\label{bound}
	\end{equation}
\end{theorem}
\begin{proof}
	See Appendix~\ref{App_G}. \phantom\qedhere
\end{proof}

\begin{remark}
	Interestingly, even though $\delta^*$ depends completely on the system parameters e.g., $\lambda_a, R_a, c_u, \theta$ and $\alpha$, we were able to limit its range only as a function of $\alpha$. The greater the $\alpha$ the smaller $\delta^*$ and consequently the greater the limitation on NOMA setup over nodes. Also, the fact that $\delta^*< 1$ is expected, means that the NOMA setup has to operate with an overall consumption power inferior to OMA setup. 
\end{remark}

\subsection{Overall Performance for $L=2$}\label{CRSc}
Differently than in the RRS case, where we were able to attain the average over all MTD success probability based only on a linear combination of $p_{j,u}^r$ as stated in \eqref{prrs}, here we cannot do the same since the weights of $p_{1,1}^{c}$ and $p_{j,2}^{c}$ when averaging are different for each $k$ and dependent on the index of the ordered channels, $i$. Instead, the average over all MTD success probability is given in \eqref{pcrs} as a function of $p_{j,u}^{c\ (i)}$  at the top of the page, where $(a)$ comes from using the expressions for $\Pr(K\le N)$, $\Pr(K=0)$ and $\Pr(K=k)$.
On the other hand, the average number of simultaneous served MTDs is given in \eqref{Kc}, shown below of \eqref{pcrs} on top of the page, where $A_1$ was given in \eqref{A1} and the approximation is because the same reasons than previously discussed. Notice that Remark~\ref{remark3} holds here for the CRS scheme as well.
\begin{figure*}[!t]
	\footnotesize
	\begin{align}
	p^c_{\mathrm{succ}}\!&=p^r_{1,1}\frac{\Pr(K\le N)}{1-\Pr(K=0)}+\frac{1}{1-\Pr(K=0)}\Bigg(\sum_{k=N+1}^{2N-1}\frac{1}{k}\bigg(\sum_{i=1}^{\mathrm{mod}(k,N)}\Big(p_{1,2}^{c\ (i)}+p_{2,2}^{c\ (i)}\Big)+\sum_{i=\mathrm{mod}(k,N)+1}^{N}p_{1,1}^{c\ (i)}\bigg)\Pr(K=k)+\nonumber\\
	&\ \ \ \ \ \ \ \ \ \ \ \ \ \ \ \ \ \ \ \ \ \ \ \ \ \ \ \ \ \ \ \ \ \ \ \ \ \ \ \ \ \ \ \ \ \ \ \ \ \ \ \ \ \ \ \ \ \ \ \ \ \ \ \ \ \ \ \ \ \ \ \ \ \ \ \ \ \ \ +\frac{1}{2N}\sum_{k=2N}^{\infty}\sum_{i=1}^{N}\Big(p_{1,2}^{c\ (i)}+p_{2,2}^{c\ (i)}\Big)\Pr(K=k)\Bigg)\nonumber\\
	&\stackrel{(a)}{\approx} p^r_{1,1}\frac{\exp(\bar{m})Q(N\!+\!1,\bar{m})\!-\!1}{\exp(\bar{m})-1}\!+\!\frac{1}{\exp(\bar{m})-1}\Bigg(\sum_{k=N+1}^{2N-1}\bigg(\sum_{i=1}^{k\!-\!N}\Big(p_{1,2}^{c\ (i)}\!+\!p_{2,2}^{c\ (i)}\Big)\!+\!\sum_{i=k\!-\!N\!+\!1}^{N}p_{1,1}^{c\ (i)}\bigg)\frac{\bar{m}^k}{k\cdot k!}\!+\nonumber\\
	&\ \ \ \ \ \ \ \ \ \ \ \ \ \ \ \ \ \ \ \ \ \ \ \ \ \ \ \ \ \ \ \ \ \ \ \ \ \ \ \ \ \ \ \ \ \ \ \ \ \ \ \ \ \ \ \ \ \ \ \ \ \ \ \ \ \ \ \ \ \ \ \ \ \  +\frac{1}{2N}\sum_{k=2N}^{k_{\max}}\sum_{i=1}^{N}\Big(p_{1,2}^{c\ (i)}+p_{2,2}^{c\ (i)}\Big)\frac{\bar{m}^k}{k!}\Bigg)\label{pcrs},\\
	\bar{K}_c\!&=\!p^r_{1,1}\sum_{k\!=\!0}^{N}k\Pr(K\!=\!k)\!+\!\sum_{k\!=\!N+\!1}^{2N\!-\!1}\!\bigg(\sum_{i\!=\!1}^{k\!-\!N}\!\Big(\!p_{1,2}^{c\ (i)}\!+\!p_{2,2}^{c\ (i)}\!\Big)\!+\!\!\!\!\!\sum_{i=k\!-N\!+\!1}^{N}\!\!p_{1,1}^{c\ (i)}\!\bigg)\!\Pr(K\!=\!k)\!+\!\sum_{k\!=\!2N}^{\infty}\!\sum_{i\!=\!1}^{N}\!\Big(p_{1,2}^{c\ (i)}\!+\!p_{2,2}^{c\ (i)}\Big)\Pr(K\!=\!k)\!\nonumber\\
	&\approx\!p^r_{1,1}A_1+\!\sum_{k\!=\!N+1}^{2N\!-\!1}\bigg(\sum_{i\!=\!1}^{k\!-\!N}\Big(p_{1,2}^{c\ (i)}\!+\!p_{2,2}^{c\ (i)}\Big)\!+\!\!\!\sum_{i=k\!-N\!+\!1}^{N}\!\!p_{1,1}^{c\ (i)}\bigg)\frac{\bar{m}^k\exp(-\bar{m})}{k!}\!+\!\sum_{k\!=\!2N}^{k_{\max}}\sum_{i\!=\!1}^{N}\Big(p_{1,2}^{c\ (i)}\!+\!p_{2,2}^{c\ (i)}\Big)\frac{\bar{m}^k\exp(-\bar{m})}{k!}\label{Kc}.
	\end{align}
	\hrule
\end{figure*}
\section{Power Consumption Analysis}\label{PC}
\begin{lemma}
	The average transmit power per orthogonal channel for the OMA ($L=1$) and, RRS and CRS of our hybrid scheme ($L=2$) is given by
	\begin{align}\label{pt}
	\bar{p}_t&=\!\left\{ \begin{array}{ll}
		(1-c_0)\Psi,& \text{for OMA}\\
		(c_1+\delta c_2)\Psi& \text{for our hybrid scheme with $L=2$}
	\end{array}
	\right.,
	\end{align}
	where $\Psi=\frac{2\rho R_a^{\alpha}}{\alpha+2}$.
\end{lemma}
\begin{proof}
	For OMA, RRS and CRS with $L=2$, the average transmit power is $\mathbb{E}[(1-c_0)\rho r^{\alpha}]=(1-c_0)\rho\mathbb{E}[r^{\alpha}]$, $\mathbb{E}[c_1\rho r^{\alpha}+2c_2\rho r^{\alpha}]=\bar{c}\rho\mathbb{E}[r^{\alpha}]$ and $\mathbb{E}[c_1\rho r^{\alpha}+c_2(a_i+b_i)\rho r^{\alpha}]=(c_1+\delta c_2)\rho\mathbb{E}[r^{\alpha}]$, respectively. For RRS $\bar{c}=c_1+\delta c_2$ since $\delta=2$, and now it is only necessary to compute $\mathbb{E}[r^{\alpha}]$,
	\begin{align}
	\mathbb{E}[r^{\alpha}]\!=\!\int\limits_{0}^{R_a}r^{\alpha}f_r(r)\mathrm{d}r\!=\!\frac{2}{R_a^2}\int\limits_{0}^{R_a}r^{\alpha+1}\mathrm{d}r\!=\!\frac{2R_a^{\alpha}}{\alpha+2}=\frac{\Psi}{\rho},
	\end{align}
	and \eqref{pt} is attained.
\end{proof}
Obviously, the RRS scheme with $L=2$ will always require a greater power consumption than an OMA setup, since whenever two MTDs are transmitting on the same channel of a cluster, the power consumption doubles. On the other hand, the CRS scheme allows to reduce the power consumption by adopting a relatively small $\delta$, e.g., $\delta^*$. This behavior is illustrated in Fig.~\ref{Fig4} for $c_0=0$. All the schemes have a common starting point in $\Psi$ since all of them are equivalent when no MTDs require to share the same orthogonal channel, e.g., $c_2=0\rightarrow c_1=1$. When $c_2$ increases, the power consumption of the RRS and CRS schemes more, as long as $\delta\ne 2$, reaching the maximum difference for $c_2=1$. Notice that choosing $\delta^*$ we guarantee operating with the same interference than an OMA setup while the power consumption reduces since $\delta^*<1$.
\begin{figure}[t!]
	\centering
	\subfigure{\includegraphics[width=0.46\textwidth]{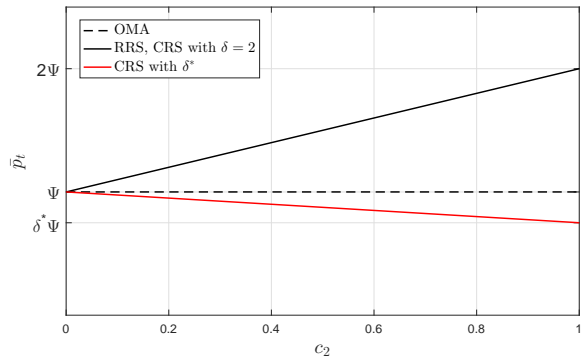}}
	\vspace{-4mm}
	\caption{Average transmit power per orthogonal channel with $c_0=0$.}		
	\label{Fig4}
	\vspace*{-4mm}
\end{figure}

\section{Numerical Results}\label{results}
Both, simulation and analytical results, are presented in this section in order to investigate the performance of our hybrid scheme as a function of the system parameters while comparing it with an OMA setup. The analytical results for the RRS scheme come from using the exact expressions; while for the CRS scheme we use the approximations. Unless stated otherwise, results are obtained by setting $\bar{m}=60$, $\lambda_a=10^{-4.4}/\mathrm{m}^2$ $(39.81/\mathrm{km}^2)$, $R_a=40$m, $\alpha=3.6$, $\mu=0$ and $\theta=1$. The value of $\delta^*$ is found by solving numerically \eqref{eq} whenever is required. Simulation results are generated using 50000 Monte Carlo runs and a sufficiently large area such 400 aggregators are placed on average.
\begin{figure}[t!]
	\centering
	\subfigure{\includegraphics[width=0.46\textwidth]{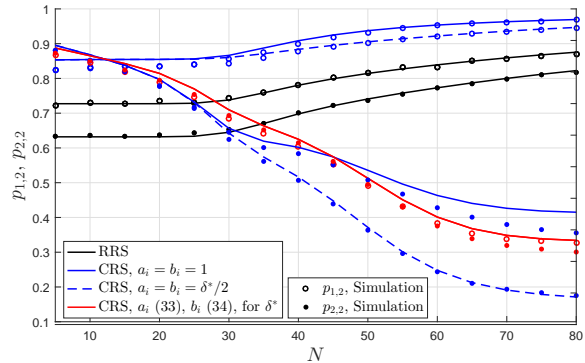}}
	\vspace{-4mm}
	\caption{Success probability of MTDs sharing the channel as a function of the number of channels.}		
	\label{Fig5}
	\vspace{-4mm}
\end{figure}

Fig.~\ref{Fig5} shows the success probability of the MTDs sharing the same channel conditioned on $u=2$, e.g., \eqref{p12} and \eqref{p22} for RRS, and \eqref{p12p22c} for CRS. The idea is to show how fair the schemes are when allocating the transmission resources. Notice that for the RRS scheme, the gap between both MTDs performance keeps quite constant, independently of $N$. This is because that gap relies strongly on the gap between $\max(h',h'')$ and $\min(h',h'')$, (see \eqref{SIR12} and \eqref{SIR22})\footnote{Notice the gap also relies strongly on the value of $\mu$.}, which only depends on the number of MTDs requiring transmissions and not on the available channels. While for the CRS schemes with fixed $a_i$, $b_i$, the gap tends to widen when increasing $N$ because $h_i$ becomes larger with respect to $h_{i+N}$ in \eqref{sir12} and \eqref{sir22}. Of course, a given system is projected to work given one value of $N$, and by properly choosing some fixed $a_i$ and $b_i$ we can reduce the performance gap between the MTDs sharing the channel, which is not possible for the RRS scheme since weighting the transmit powers is not available. Notice also that by choosing $a_i$ and $b_i$ according to \eqref{ai} and \eqref{bi} both MTDs reach similar performance, hence the fairest scheme. 
\begin{figure}[t!]
	\centering
	\subfigure{\label{Fig6a}\includegraphics[width=0.46\textwidth]{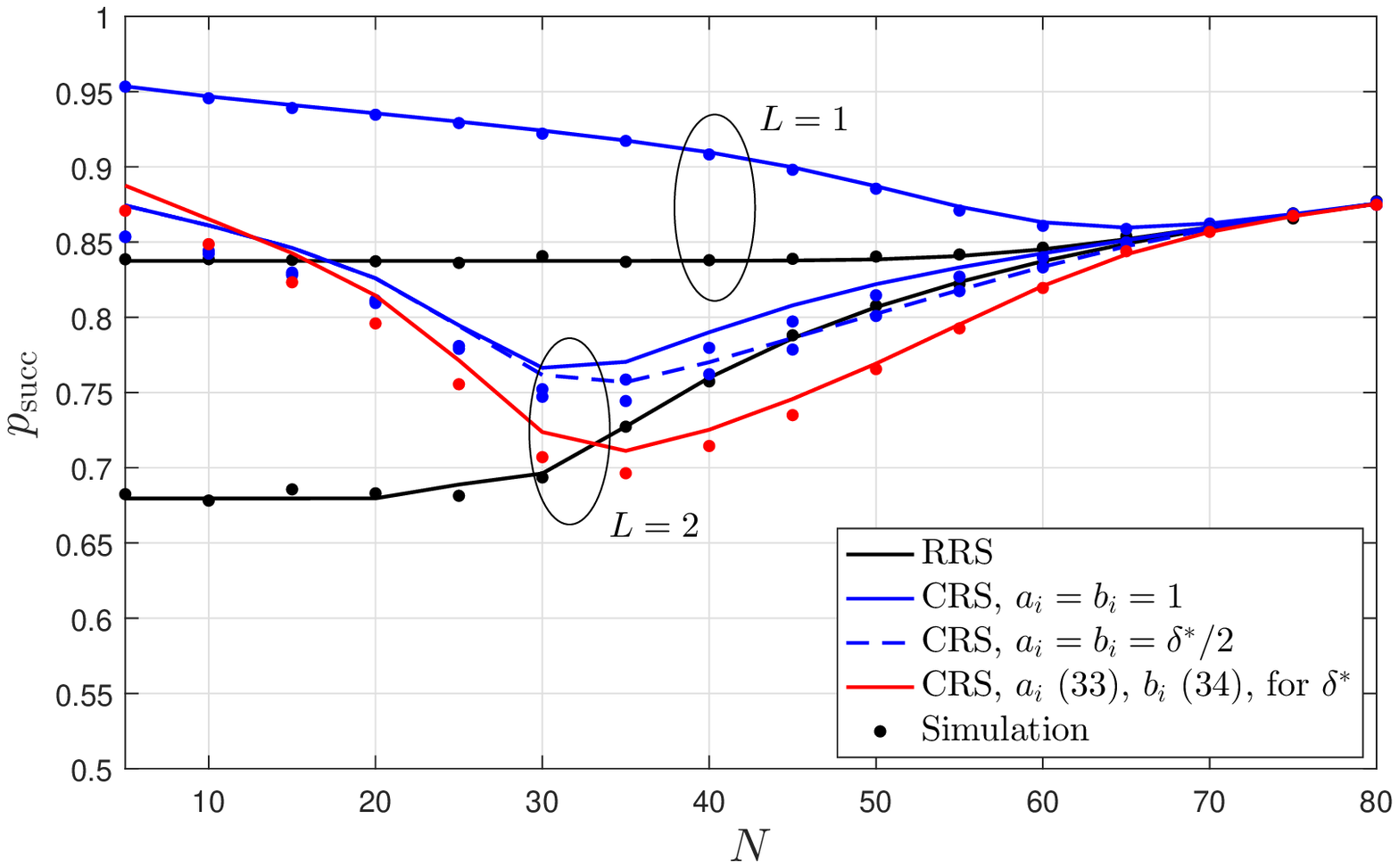}}\ \ \ \ 
	\subfigure{\label{Fig6b}\includegraphics[width=0.46\textwidth]{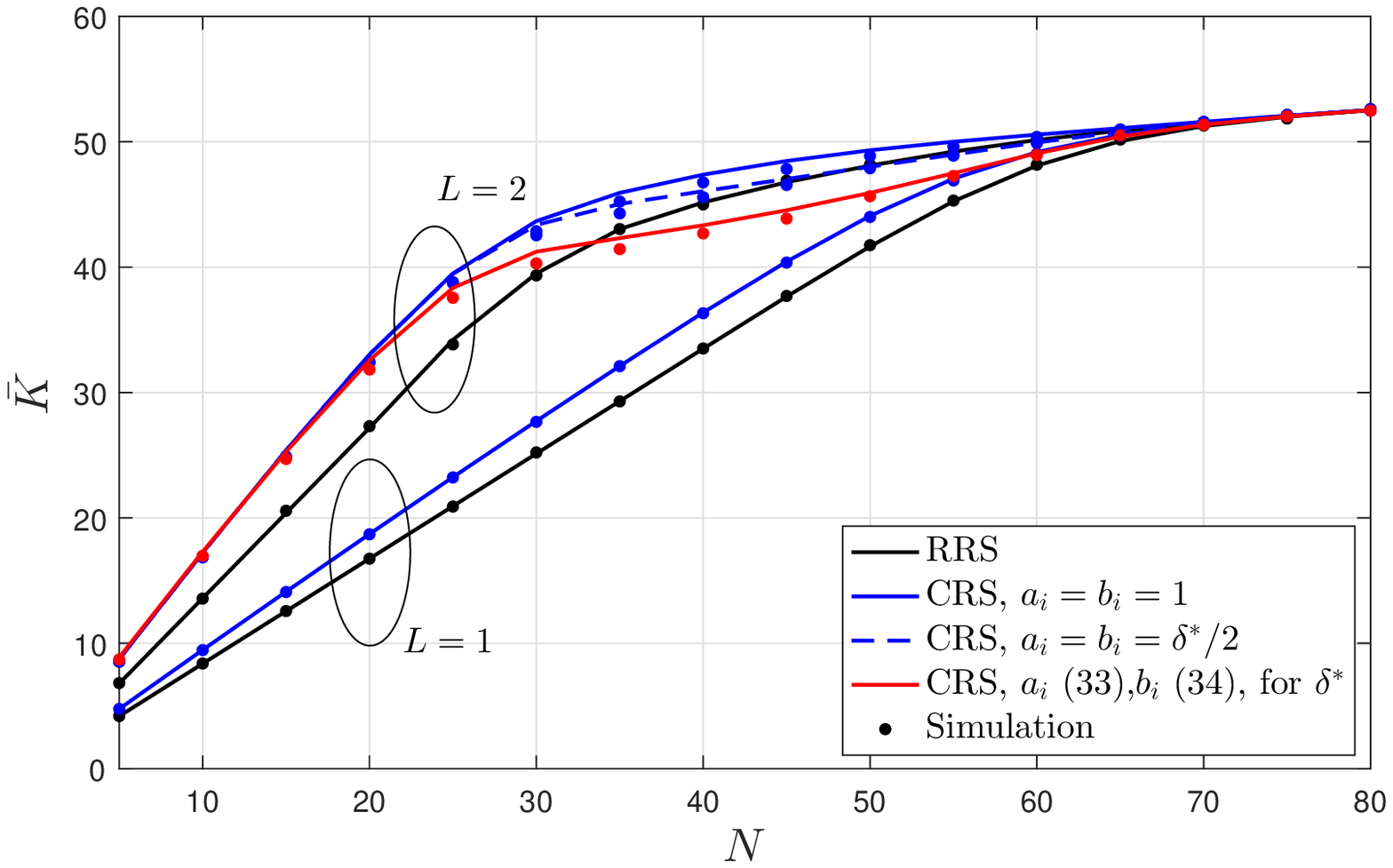}}	
	\vspace{-4mm}
	\caption{(a) Average over all MTD success probability (left) and (b) Average number of simultaneously served MTDs (right), as a function of the number of channels for $L=\{1,2\}$.}\label{Fig6}			
	\vspace{-4mm}
\end{figure}

Fig.~\ref{Fig6}a shows the average over all MTD success probability for the different schemes and $L=\{1,2\}$. Each scheme with $L=1$ performs better than for $L=2$. Also, as expected and in general, the CRS setups outperforms the RRS scheme, except when choosing $a_i,\ b_i$ according to \eqref{ai} and \eqref{bi}, for which a greater fairness is attained instead. Notice that when $N$ is small, CRS outperforms RRS; while as $N$ increases, their curves tend to overlap. This is due to the fact that when $N$ is large,
	\textit{i)} e.g, greater than $\bar{m}$ for $L=1$, most of the time the number of MTDs is less than the available resources such that the implementation of the CRS is almost the same as the RRS; 
\textit{ii)} e.g., greater than $\bar{m}/2$ for $L=2$, the impact of ordering the channels becomes less significant since most of the time all the MTDs will be allocated.  While when $N$ increases more and more, the system tends to behave as when $L=1$.
The change in the CRS curves from decreasing to increasing  occurring close to $\bar{m}/2$ is somewhat explained by that latter argument. Even when achieving a higher reliability with an $L=2$ setup is not possible, we are able to enhance the average number of simultaneously served MTDs, as shown in Fig.~\ref{Fig6}b. The success probability does not deteriorate and a significant improvement on $\bar{K}$ is attained, fundamentally when $N$ is not large. Notice that this advantage is evinced when $L=2$, especially with the CRS scheme, so to cover a high instantaneous access demand $\bar{m}>N$, which is even more favorable than predicted in Remark~\ref{remark3}. By setting $L=2$, the number of efficiently served MTDs could be up to twice the number with a single MTD per orthogonal channel setup. Particularly attractive are the configurations operating with $\delta^*$ since no change in the perceived interference occurs when switching transmission schemes from $L=1$ and $L=2$ setups and vice versa.

\begin{figure}[t!]
	\centering
	\subfigure{\label{fig:a}\includegraphics[width=0.46\textwidth]{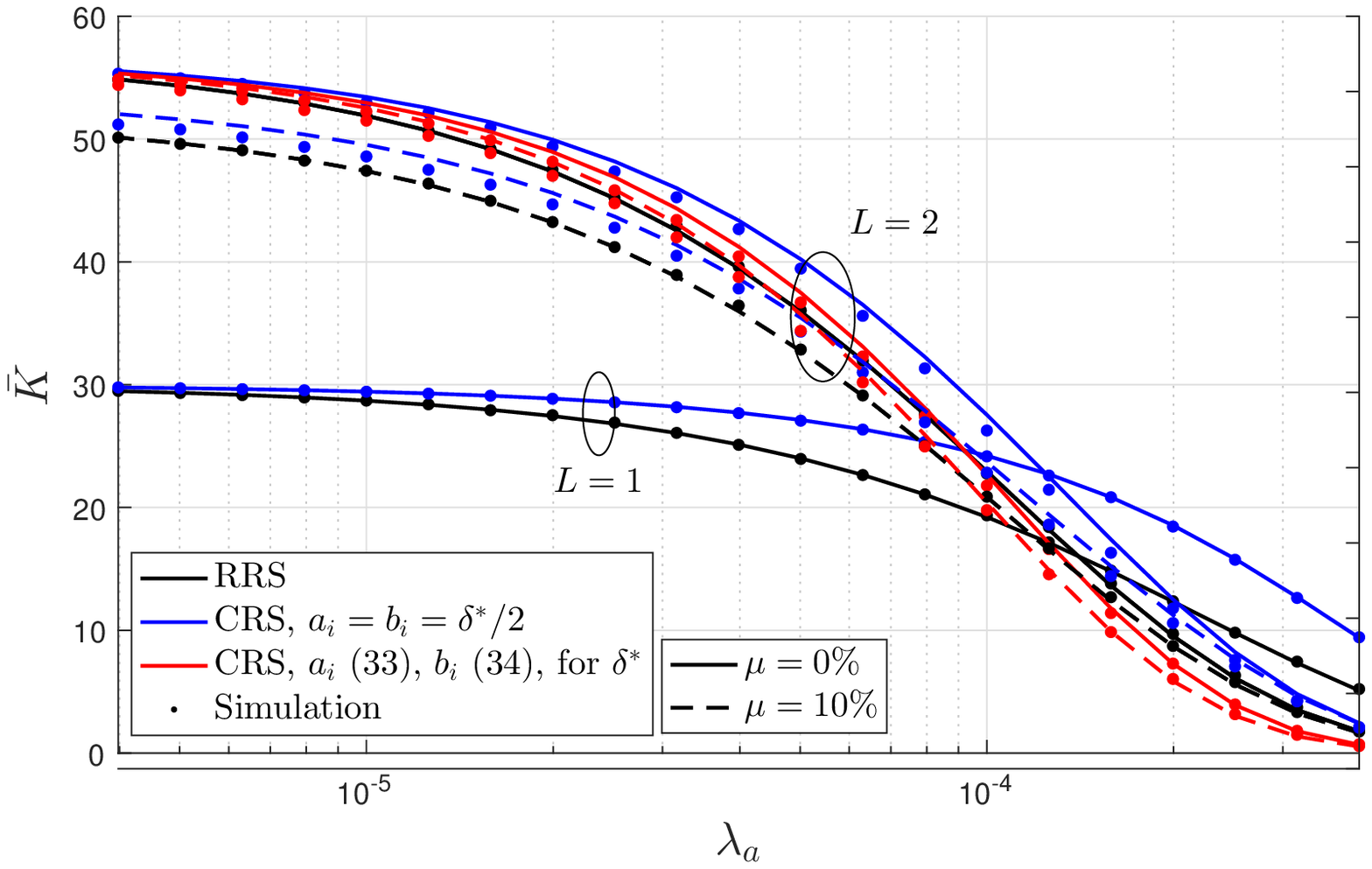}}\ \ \ \ 
	\subfigure{\label{fig:b}\includegraphics[width=0.46\textwidth]{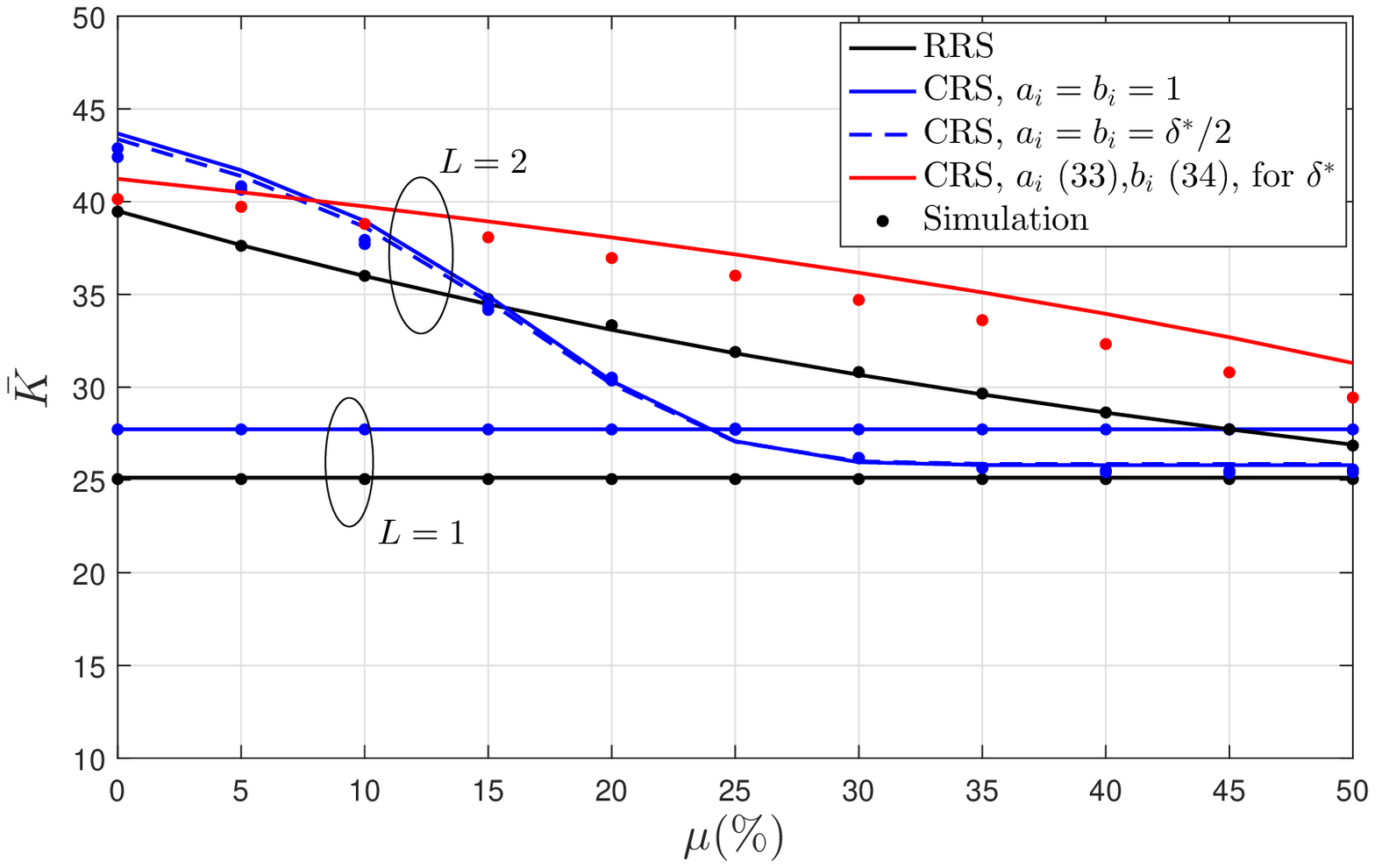}}	
	\vspace{-4mm}
	\caption{Average number of simultaneously served MTDs, as a function of: (a) density of aggregators (left), (b) SIC imperfection coefficient (right), for $N=30$.}	
	\label{Fig7}	
	\vspace{-4mm}
\end{figure}
Fig.~\ref{Fig7} shows the average number of simultaneously served MTDs as a function of the density of aggregators for $\mu\in\{0,10\%\}$, Fig.~\ref{Fig7}a, and the SIC imperfection coefficient, Fig.~\ref{Fig7}b. When both, the network is sparse and the SIC imperfection coefficient are not so restrictive, the performance of the $L=2$ setup increases. Since $N=30$, each channel per cluster is operating with two MTDs almost all the time, which are more sensitive to the interference, hence their performance will be affected if either $\lambda_a$, $\theta$, or even $R_a$ are larger. If $c_2$ was smaller, e.g., larger $N$, this situation becomes less critical, although the gap between the $L=1$ and $L=2$ setups would be smaller as shown previously in Fig.~\ref{Fig6}. Notice that the imperfect SIC degrades the system performance when $L=2$, but even for a high imperfection such $10\%$, the advantage over the OMA setup keeps evident for a wide range of values for the system parameters. For instance, the CRS scheme with $a_i=b_i=\delta^*/2$ overcomes the CRS scheme with $L=1$ when $\lambda_a\lesssim 1.3\cdot 10^{-4}$ for $\mu=0$, while $\lambda_a\lesssim 1\cdot 10^{-4}$ would be required for $\mu=10\%$. Since SIC is only related with the $L=2$ setup, the OMA setup appears shown as a straight line for both RRS and CRS schemes in Fig.~\ref{Fig7}b. Notice that the setup with fixed power coefficients, e.g., RRS and CRS with fixed $a_i,\ b_i$, are the most affected when $\mu$ increases since those coefficients work well for certain system parameters but others will be required if they change, e.g. different $\mu$ in this case. It is expected that a smaller $a_i$, hence larger $b_i$, work better as $\mu$ increases (see \eqref{sir12} and \eqref{sir22}). The setup with $a_i$ and $b_i$ given by \eqref{ai} and \eqref{bi}, respectively, adapts better to the different $\mu$s since this is a parameter taken into account when performing their calculations. This is, the values of $a_i$ and $b_i$ that allow a similar reliability between both MTDs sharing the same channel in a given cluster, depend on $\mu$, thus no additional adjustment on them is required. It is clear that failing to efficiently  eliminate the intra-cluster interference could reduce significantly the benefits from NOMA, and can be a  challenging issue for implementing NOMA in practice.

Finally, notice that simulations and analytical expressions, even when several of them are approximations, fit well in all the cases, e.g., Fig.~\ref{Fig5}-\ref{Fig7}, which validates our findings\footnote{Analytical expressions, even when some of them are numerically evaluated, have an additional value since simulations of these scenarios require a huge amount of computational resources specially for $L>1$.}.
\section{Conclusion}\label{conclusions}
In this paper, we analyzed the uplink mMTC in a large-scale cellular network system overlaid with data aggregators. We propose a hybrid access scheme, OMA-NOMA, while developing a general analytical framework to investigate its performance in terms of average success probability and average number of simultaneously served MTDs for two scheduling schemes RRS and CRS. Also, we found the power constraints on the MTDs sharing the same channel to attain a fair coexistence with purely OMA setups, while power control coefficients are found too, so that both MTDs can perform with similar reliability.  Our analytical derivations allow for fast computation compared to the time-consuming Monte-Carlo simulations, which are even heavier for our hybrid scheme than for a purely OMA setup. The numerical results show that
\begin{itemize}
	\item our hybrid access scheme aims at providing massive connectivity in scenarios with high access demand, which is not covered by traditional OMA setups, and even with lower average power consumption per orthogonal channel and per MTD, the hybrid scheme with CRS outperforms the OMA setup; 
	\item failing to efficiently  eliminate the intra-cluster interference could reduce significantly the benefits from NOMA, and can be a  challenging issue for implementing NOMA in practice;
	\item our mathematical derivations, besides being easy to evaluate, are accurate.	
\end{itemize}

Future work could focus on the relaying phase, while finding some strategies to cope with the larger aggregated data. Additionally, we intend to deeply investigate strategies in order to optimally decide when to switch from a purely OMA setup to our hybrid scheme.

\appendices 
\section{Proof of Lemma~\ref{lemma1}}\label{App_A}
The PMF of the number of MTDs sharing the same channel, $u$, conditioned on the number of MTDs requiring transmissions, $k$, is given by
\small
\begin{align}
\Pr(U=u|k)=&\!\left\{ \begin{array}{ll}
1,&\! \mathrm{for}\ u=L\ \mathrm{if}\ k\!\ge\! N L\\
1\!-\!\frac{k}{N}\!+\!\lfloor\tfrac{k}{N}\rfloor\!,&\! \mathrm{for}\ u\!=\!\lfloor\tfrac{k}{N}\rfloor\ \mathrm{if}\ k\!<\!N L \\
\frac{k}{N}\!-\!\lfloor\tfrac{k}{N}\rfloor\!,&\! \mathrm{for}\  u\!=\!\lfloor\tfrac{k}{N}\rfloor\!+\!1\!=\!\lceil\tfrac{k}{N}\rceil\ \mathrm{if}\ k\!<\!N L \\
0,& \mathrm{otherwise}
\end{array}\!.
\right.\label{pu}
\end{align}
\normalsize
Notice that if $k$ is equal or greater than the number of available resources, $NL$, there will be for sure $u=L$ MTDS per channel in the representative cluster. Otherwise, if $k<NL$, only two consecutive values for $u$ are possible. For instance, if $N=10$, $L=4$ and $k=26$, then 2 MTDs will be allocated in each of 4 channels ($2\times 4=8$ MTDs), and 3 MTDs in the remaining $6$ channels ($3\times 6=18$ MTDs). Therefore, the probability of one channel being occupied by 2 MTDs is $4/10=0.4=1-26/10+\lfloor 26/10 \rfloor$, while the probability of being occupied by 3 is the complement, $0.6=26/10-\lfloor 26/10 \rfloor$.

Now, the required PMF, $\Pr(U=u)=\Pr(U=u|k)\Pr(K=k)$, can be calculated as follows
\small
\begin{align}
\Pr(U=u)&=\!\left\{ \begin{array}{ll}
\sum\limits_{k=NL}^{\infty}\Pr(K=k),&\! \mathrm{for}\ u\!=\!L\ \\
\Big(1\!-\!\frac{k}{N}\!+\!\lfloor\tfrac{k}{N}\rfloor\Big)\sum\limits_{k=0}^{NL\!-\!1}\Pr(K\!=\!k),&\! \mathrm{for}\ u\!=\!\lfloor\tfrac{k}{N}\rfloor\ \\
\Big(\frac{k}{N}\!-\!\lfloor\tfrac{k}{N}\rfloor\Big)\sum\limits_{k\!=\!0}^{NL\!-\!1}\Pr(K\!=\!k),&\! \mathrm{for}\  u\!=\!\lfloor\tfrac{k}{N}\rfloor\!+\!1\!=\!\lceil\tfrac{k}{N}\rceil\ \\
0,& \mathrm{otherwise}
\end{array},
\right.\nonumber\\
&=\!\left\{ \begin{array}{l}
\sum\limits_{k\!=\!0}^{N\!-\!1}\Big(1\!-\!\tfrac{k}{N}\Big)\Pr(K\!=\!k),\ \ \ \ \ \ \ \ \ \ \ \ \ \ \ \ \ \ \ \ u\!=\!0 \\
\sum\limits_{k\!=\!Nu}^{N(u\!+\!1)\!-\!1}\Big(1\!-\!\tfrac{k}{N}\!+\!u\Big)\Pr(K\!=\!k)\!+\!\sum\limits_{k\!=\!N(u\!-\!1)}^{Nu\!-\!1}\Big(\tfrac{k}{N}\!-\!u\!+\!1\Big)\Pr(K\!=\!k),\ \ \ u\!=\!1,...,L\!-\!1 \\
\ \sum\limits_{k\!=\!NL}^{\infty}\Pr(K\!=\!k)\!+\! \sum\limits_{k\!=\!N(L\!-\!1)}^{NL\!-\!1}\Big(\tfrac{k}{N}\!-\!L\!+\!1\Big)\Pr(K\!=\!k), u\!=\!L \\
0,\ \ \ \ \ \ \ \ \ \ \ \ \ \ \ \ \ \ \ \ \ \ \ \ \ \ \ \ \ \ \ \ \ \ \ \ \ \ \ \ \ \ \ \ \ \ \ \mathrm{otherwise} 
\end{array}\!\!.
\right.
\end{align}
\normalsize
Using the PDF of $K$, $\Pr(K=k)$, its CDF, $\Pr(K\le k)$, and the expected $K$ on one interval,
\begin{align}
&\sum_{k=a}^{k=b}k\Pr(K=k)=\!\bar{m}\Big(Q(a,\bar{m})\!+\!Q(1\!+\!b,\bar{m})\Big)\!+\!\exp(\!-\!\bar{m})\Big(\frac{\bar{m}^a}{(a\!-\!1)!}\!-\!\frac{\bar{m}^{1\!+\!b}}{b!}\Big)\nonumber\\
&\stackrel{a=0}{=}\bar{m}Q(1+b,\bar{m})-\frac{\exp(-\bar{m})\bar{m}^{1+b}}{b!}\stackrel{b\rightarrow\infty}{=}\bar{m}\big(1-Q(a,\bar{m})\big)+\frac{\exp(-\bar{m})\bar{m}^a}{(a-1)!}\nonumber,
\end{align}
{and regrouping similar terms, }we attain \eqref{Ueq}. \hfill 	\qedsymbol
\section{Proof of Theorem~\ref{the1}}\label{App_B}
We proceed as follows
\begin{align}
p^r_{1,1}&\!=\!\Pr\big(\mathrm{{SIR}_{1,1}}\!>\!\theta\big)\!=\!\mathbb{E}_I[\Pr(h\!>\!\theta I|I)]\!=\!\mathbb{E}_I\Big[\exp(\!-\!\theta I)\Big|I\Big]\!,\!\label{Bp11}\\
p^r_{1,2}\!&=\Pr\big(\mathrm{{SIR}_{1,2}}>\theta\big)
=\mathbb{E}_I\Big[\Pr\big(\max(h',h'')-\theta\min(h',h'')>\theta I|I\big)\Big]
=\mathbb{E}_I\Big[\Pr\big(v_1>\theta I|I\big)\Big]\nonumber\\
&\!\stackrel{(a)}{=}\!\left\{ \begin{array}{ll}
\mathbb{E}\Big[\frac{2}{1\!+\!\theta}\exp(-\theta I)\!-\!\frac{1\!-\!\theta}{1\!+\!\theta}\exp(-\frac{2\theta}{1\!-\!\theta} I)\Big|I\Big]\!,\!& \mathrm{if}\ 0\!\le\!\theta\!<\!1\!\\
\mathbb{E}\Big[\frac{2}{1+\theta}\exp(-\theta I)|I\Big],& \mathrm{if}\ \theta\ge 1
\end{array}
\right.\!,\label{Bp12}
\end{align}
\begin{align}
p^r_{2,2}\!&=\Pr\big(\mathrm{{SIR}_{2,2}}>\theta\big)
=\mathbb{E}_I\Big[\Pr\big(\min(h',h'')-\theta\mu\max(h',h'')>\theta I|I\big)\Big]
=\mathbb{E}_I\Big[\Pr\big(v_2>\theta I|I\big)\Big]\nonumber\\
&\stackrel{(b)}{=}\!\left\{ \begin{array}{ll}
\mathbb{E}\Big[\frac{1-\theta\mu}{1+\theta\mu}\exp(-\frac{2\theta}{1-\theta\mu}I)\Big|I\Big]\!,\!& \mathrm{if}\ 0\!\le\!\theta\mu\!<\!1\!\\
0,& \mathrm{if}\ \theta\mu\ge 1
\end{array}
\right.\!\label{Bp22},
\end{align}
where $v_1=\max(h',h'')-\theta\min(h',h'')$ and $v_2=\min(h',h'')-\theta\mu\max(h',h'')$, while $(a)$ and $(b)$ come from using their CDF expressions, which are given by
\begin{align}
F_{V_1}(v_1)&=\left\{ \begin{array}{ll}\!1-
\frac{2}{1\!+\!\theta}\exp(-v_1)\!+\!\frac{1\!-\!\theta}{1\!+\!\theta}\exp(-\frac{2v_1}{1\!-\!\theta}),& \mathrm{if}\ 0\le\!\theta\!<1\\
1-\frac{2}{1+\theta}\exp(-v_1),& \mathrm{if}\ \theta\ge 1\\
\end{array}
\right.\!,\label{Fv1}\\
F_{V_2}(v_2)&=\left\{ \begin{array}{ll}1-
\frac{1-\theta\mu}{1+\theta\mu}\exp(-\frac{2v_2}{1-\theta\mu}),& \mathrm{if}\ 0\le\theta\mu<1\\
1,& \mathrm{if}\ \theta\mu\ge 1\\
\end{array}
\right.\!,\label{Fv2}
\end{align}
for $v_1,v_2>0$. Notice that the last equalities in \eqref{Bp11}, \eqref{Bp12} and \eqref{Bp22} are equivalent to \eqref{p11}, \eqref{p12} and \eqref{p22}, respectively.

For the derivation of \eqref{LI} we have to use the fact that the Poisson cluster process defined in \eqref{phin} is a Neyman-Scott process \cite[Definition 3.4]{Haenggi.2012} with Probability Generating Functional (PGFL) \cite[Corollary 4.13]{Haenggi.2012}
	\begin{align}
	G[\upsilon]\!=\!\exp\!\bigg(\!\lambda_a\!\!\int\limits_{\mathbb{R}^2}\!\!\Big(G_0\Big(\!\int\limits_{\mathbb{R}^2}\!v(\mathbf{w}\!+\!y)f_y(y)\mathrm{d}y\Big)\!-\!1\Big)\mathrm{d}\mathbf{w}\!\!\bigg).\label{gpcp}
	\end{align}	
Based on $g\sim\mathrm{Exp}(1)$ which allows to state
\begin{align}
v(\mathbf{w}+y)&=\mathbb{E}_g\big[\exp(-sgr_a^{\alpha}||\mathbf{w}+y||^{-\alpha})\big]=\!\frac{1}{1+sr_a^{\alpha}||\mathbf{w}+y||^{-\alpha}},
\end{align}
and according to the cosine law, e.g., $||\mathbf{w}+y||=(r_{\mathbf{w}}^2+r_a^2-2r_{\mathbf{w}}r_a\cos(\omega))^{\frac{1}{2}}$, while substituting \eqref{M} and PDFs expressions $f_{r_a}(r_a)=\tfrac{2r_a}{R_a^2}$ and $f_{\omega}(\omega)=\tfrac{1}{2\pi}$ into \eqref{gpcp}, we reach $\mathcal{L}_{I_r}(s)$ in \eqref{LI}. \hfill 	\qedsymbol

\section{Proof of Theorem~\ref{lemma2}}\label{New}
Performing some algebraic transformations on \eqref{LI} we have
\small 
\begin{align}
\mathcal{L}_{I_r}(s)&=\exp\bigg(2\pi\lambda_a\int\limits_{0}^{\infty}r_{\mathbf{w}}\Big(\sum\limits_{u=0}^{L}c_u\Upsilon(r_{\mathbf{w}},s)^u-1\Big)\mathrm{d}r_{\mathbf{w}}\bigg)
\nonumber\\
&\!\stackrel{(a)}{=}\!\exp\!\bigg(\!2\pi\lambda_a\!\Big(\!\int\limits_{0}^{\infty}\!r_{\mathbf{w}}\Big(\sum\limits_{u=0}^{L}\big(c_u\Upsilon(r_{\mathbf{w}},s)^u\!-\!c_u\big)\Big)\mathrm{d}r_{\mathbf{w}}\!\Big)\!\bigg)\nonumber\\
&\!\stackrel{(b)}{=}\!\exp\!\bigg(2\pi\lambda_a\sum\limits_{u=1}^{L}c_u\Big(\int\limits_{0}^{\infty}r_{\mathbf{w}}\big(\Upsilon(r_{\mathbf{w}},s)^u-1\big)\mathrm{d}r_{\mathbf{w}}\Big)\!\bigg)\nonumber\\
&\!\stackrel{(c)}{=}\!\underbrace{\exp\bigg(2\pi c_1\lambda_a\int\limits_{0}^{\infty}\!r_{\mathbf{w}}\big(\Upsilon(r_{\mathbf{w}},s)\!-\!1\big)\mathrm{d}r_{\mathbf{w}}\bigg)}_{T_1}\underbrace{\exp\bigg(2\pi\lambda_a\!\!\sum\limits_{u=2}^{L}\!c_u\Big(\int\limits_{0}^{\infty}\!r_{\mathbf{w}}\big(\Upsilon(r_{\mathbf{w}},s)^u\!-\!1\big)\mathrm{d}r_{\mathbf{w}}\!\Big)\bigg)}_{T_2}\!,\label{LirAp}
\end{align}
\normalsize
where $(a)$ comes from $\sum_{u=0}^{L}c_u=1$, $(b)$ is attained by regrouping terms, and $(c)$ by pulling out the term associated with $u=1$. Notice that $T_1$ matches the Laplace transform of an HPPP with density $c_1\lambda_a$ and one active MTD per channel per cluster, which is given by \cite[Eq. (23)]{Guo.2017}
\begin{align}
	T_1=\exp(-\chi c_1s^{\frac{2}{\alpha}})\label{T1}
\end{align}
On the other hand, $T_2$ includes the contribution of the clustered MTDs, e.g., $u\ge 2$. For each $u$, the related term in $T_2$ matches the Laplace transform of an HPPP with density $c_u\lambda_a$ and $u$ active MTDs per channel per cluster. It has been observed in \cite{Haenggi.2014,Guo.2015,Ganti.2016} that the SIR complementary CDFs, e.g., Laplace transform of the interference, for different point processes appear to be merely horizontally shifted versions of each other (in dB), as long as their diversity gain is the same. Thus, scaling the threshold $s$ by this SIR gain factor (or shift in dB) $G$, we have
\begin{align}
\mathcal{L}_I(s)\approx\mathcal{L}_{I,\mathrm{ref}}(s/G).
\end{align}
However, $G$ is also a function of $s$ but for many setups it keeps approximately constant and consequently it can be determined by finding its value for an arbitrary value of $s$ \cite{Ganti.2016}.
Using the PPP as the reference model, the limit of $G$ as $s\rightarrow 0$, $G_0$, is relatively easy to calculate \cite{Ganti.2016}
\begin{align}
G_0=\frac{\mathrm{MISR_{PPP}}}{\mathrm{MISR}},
\end{align} 
where the MISR is the mean of the interference-to-(average)-signal ratio $\mathrm{I\bar{S}R}=\tfrac{I}{\mathbb{E}_h[S]}$. Since $S=h$ here and $\mathbb{E}[h]=1$, we have that $\mathrm{MISR}=\mathbb{E}[I]$. Now, considering the contribution of each $u$ in $T_2$ separately we have
\begin{align}
G_{0,u}=\frac{\mathbb{E}[I_{\mathrm{PPP}}]}{\mathbb{E}[I_r]}=\frac{1}{u}.\label{G0}
\end{align} 
Notice that $\mathbb{E}[I_{\mathrm{PPP}}]$ and $\mathbb{E}[I_r]$ are divergent measures for our system model because the resultant point process would no be locally finite since we assumed a path loss model $||x||^{-\alpha}$ \cite{Haenggi.2012}. However, the quotient depends merely on the density of both process and since $\lambda_{\mathrm{PPP}}=c_u\lambda_a$ and $\lambda=uc_u\lambda_a$ we attain the last equality in \eqref{G0}. Obviously, $\mathcal{L}_{I,\mathrm{PPP}}(us)$ works \textit{almost surely}\footnote{Except for relatively large $s$ \eqref{up} might serve as an upper bound but it becomes an accurate approximation for \eqref{LI}.} as an upper bound for $\mathcal{L}_{I_r}(s)$ in \eqref{LI} if $G_0<1$ (see \cite{Ganti.2016} for a geometrical perspective), which holds here since $u\ge 2$.
  Now, using \cite[Eq.~(23)]{Guo.2017} as the HPPP of reference with one fixed MTD per orthogonal channel\footnote{The clustered process with one MTD per orthogonal channel (per cluster) is indeed a HPPP with density $\lambda_a$ because the displacement theorem in stochastic geometry\cite{Haenggi.2012}.}, we attain 
  \begin{align}
  T_2\stackrel{a.s}{\le}\exp\Big(-\chi \sum\limits_{u=2}^{L}c_u (us)^{\frac{2}{\alpha}}\Big)\label{T2up}
  \end{align}
  with asymptotic equality as $s\rightarrow 0$. Substituting \eqref{T2up} and \eqref{T1} into \eqref{LirAp} yields \eqref{up}. 

On the other hand, the lower bound comes from the corresponding HPPP with same intensity that for each $u$ in $T_2$, e.g., $uc_u\lambda_a$. This is because in our system model the performance in terms of success probability of an HPPP with the same intensity that any clustered process will be always worse since we are expecting a greater number of close interfering nodes. Once again using \cite[Eq.~(23)]{Guo.2017}, we have
\begin{align}
T_2\ge \exp\Big(-\chi \sum\limits_{u=2}^{L}uc_us^{\frac{2}{\alpha}}\Big).\label{T2lo}
\end{align}
Now, substituting \eqref{T2lo} and \eqref{T1} into \eqref{LirAp} yields \eqref{lo}. 
For both \eqref{up} and \eqref{lo} the equality fits to the asymptotic cases $\lambda_a\rightarrow\infty,0$, $R_a\rightarrow\infty,0$.  \hfill 	\qedsymbol
\section{Proof of Theorem~\ref{the3}}\label{App_D}
Let's assume the case where $K>N$ since for $K\le N$ the system behaves exactly as in the RRS scheme and the problem is already solved, e.g., $p^{c\ (i)}_{j,u}=p^r_{j,u}$. According to high order statistic theory \cite{Herbert.2006}, the PDF of the $i$th best channel power gain is
\begin{align}
f_{H_i,K}(x)&\!=\!\frac{K!\exp(-ix)\big(1\!-\!\exp(-x)\big)^{K-i}}{(i-1)!(K-i)!}.\label{pdfi}
\end{align}
Now we have
%
$p^{c\ (i)}_{j,u}=\Pr\big(\mathrm{SNR}^{c\ (i)}_{j,u}>\theta\big)=\Pr\big(\Theta^i_{j,u}>\theta I_c\big)$,
%
where $\Theta^i_{1,1}=h_i$, $\Theta^i_{1,2}=a_ih_i-\theta b_ih_{i+N}$ and $\Theta^i_{2,2}=b_ih_{i+N}-\theta\mu a_ih_i$. 
Unfortunately, the distribution of $\Theta^i_{j,u}$, even for the simplest case $\Theta^i_{1,1}$, conduces to very complicated expression for the CDF, preventing us to follow the same path we used for the RRS scheme. Instead we proceed as follows
	\begin{align}
	p^{c\ (i)}_{j,u}&=\Pr\Big(I_c<\frac{\Theta^i_{j,u}}{\theta}\Big)\stackrel{(a)}{=}\frac{1}{2}\!-\!\frac{1}{\pi}\int\limits_{0}^{\infty}\frac{1}{\varphi}\mathrm{Im}\bigg\{\mathcal{L}_{I_c}(-\mathbf{i}\varphi)\mathbb{E}_{\Theta^i_{j,u}}\Big[\!\exp\!\Big(\!-\!\mathbf{i}\varphi\frac{\Theta^i_{j,u}}{\theta}\!\Big)\Big]\bigg\}\mathrm{d}\varphi\nonumber\\
	&\!\stackrel{(b)}{\approx}\!\frac{1}{2}\!-\!\frac{1}{\pi}\!\!\int\limits_{0}^{\infty}\!\!\frac{1}{\varphi}\mathrm{Im}\bigg\{\!\mathcal{L}_{I_c}(\!-\mathbf{i}\varphi)\exp\!\Big(\!\!\!-\!\mathbf{i}\varphi\frac{\mathbb{E}\big[\Theta^i_{j,u}\big]}{\theta}\!\Big)\!\!\bigg\}\mathrm{d}\varphi,\label{pcApD}
	\end{align}
	where $(a)$ comes from the Gil-Pelaez inversion theorem \cite{Renzo.2014} and the approximation in $(b)$ comes from the Jensen inequality. Also,
	\small
	\begin{align}
	\mathbb{E}\big[\Theta^i_{1,1}\big]&=\mathbb{E}[h_i]=\int_{0}^{\infty}xf_{H_i,K}(x)\mathrm{d}x
	=\frac{K!}{(i\!-\!1)!(K\!-\!i)!}\int\limits_{0}^{\infty}\!x\exp(-ix)\big(1\!-\!\exp(-x)\big)^{K\!-\!i}\mathrm{d}x\nonumber\\
	&\!=\!\frac{(\!-\!1)^{K\!-\!i}K!}{K^2(i\!-\!1)!(K\!-\!i)!}\!\bigg(\!K^2x\mathrm{Beta}\!\Big(\!\!\exp(x),\!-\!K,\!K\!-\!i\!+\!1\!\Big)\! \nonumber\\
	 &\qquad\qquad\qquad -\!\exp(\!-\!Kx) _3F_2\Big(\!\!-\!K,\!-\!K,i\!-\!K,1\!-\!K,1\!-\!K,\exp(x)\Big)\!\bigg)\Bigg|_{x=0}^{x=\infty}\nonumber\\
	&=\frac{\Gamma(i)\Gamma(K-i+1)\big(\psi(K+1)-\psi(i)\big)}{(i-1)!(K-i)!}=\mathrm{\psi(K+1)}-\mathrm{\psi(i)},\label{psihi}
	\end{align}
	\normalsize
	while it is straightforward obtaining $\mathbb{E}\big[\Theta^i_{1,2}\big]$ and $\mathbb{E}\big[\Theta^i_{2,2}\big]$ from \eqref{psihi}, yielding \eqref{B11}-\eqref{B22}. Substituting them into \eqref{pcApD} we attain \eqref{pc}.

	To attain \eqref{LI2} we require to use the same arguments than previously discussed when deriving the result in \eqref{LI}. However, now the point process has marks $1,a_i,b_i$, with $b_i=\delta-a_i$, and we require to include the marks along with their probabilities when evaluating \eqref{gpcp} \cite[Th. 7.5]{Haenggi.2012}.   
		Notice that it is intractable evaluating \eqref{LI2} efficiently when $a_i,b_i$ are random since requires an additional integration\footnote{For instance, the uniform distribution, which is very simple and probably unrealistic for the scenario discussed, with PDF $f_{a_i}(a_i)=\tfrac{1}{\delta}$, is already cumbersome when evaluating \eqref{LI2}.} in the exponent, hence we propose using some approximations as follows		%
\small
\begin{align}
\mathcal{L}_{I_c}(s)\!&\!\stackrel{(a)}{\approx}\!\exp\!\bigg(\!2\pi\lambda_a\!\int\limits_{0}^{\infty}\!r_{\mathbf{w}}\!\Big(\!c_0\!+\!c_1\!\Upsilon(r_{\mathbf{w}},s)\!+\!\frac{c_2}{2}\big(\Upsilon(r_{\mathbf{w}},a_is)^2+\Upsilon(r_{\mathbf{w}},b_is)^2\big)-1\Big)\mathrm{d}r_{\mathbf{w}}\bigg)\!\nonumber\\
\!&\stackrel{(b)}{=}\!\exp\!\bigg(\!2\pi c_1\lambda_a\!\!\int\limits_{0}^{\infty}\!\!\big(\Upsilon(r_{\mathbf{w}},s)\!-\!1\big)\!\mathrm{d}r_{\mathbf{w}}\!\!\bigg)\!\exp\!\bigg(\!2\pi \frac{c_2}{2}\lambda_a\!\!\int\limits_{0}^{\infty}\!\!\big(\Upsilon(r_{\mathbf{w}},a_is)^2\!-\!1\big)\mathrm{d}r_{\mathbf{w}}\!\!\bigg)\!\exp\!\bigg(\!2\pi \frac{c_2}{2}\lambda_a\!\!\int\limits_{0}^{\infty}\!\!\big(\Upsilon(r_{\mathbf{w}},b_is)^2\!-\!1\big)\mathrm{d}r_{\mathbf{w}}\!\!\bigg),\label{Ica}
\end{align}
\normalsize
where $(a)$ comes from using fixed values of $a_i,b_i$, thus, avoiding the additional integration, and using the relation between the geometric and arithmetic mean, $\tfrac{x+y}{2}\ge \sqrt{xy}\rightarrow xy\le \tfrac{x^2+y^2}{2}$, as an approximation, with equality when $a_i=b_i$; while $(b)$ comes from regrouping terms and $\sum_{u=0}^{2}c_u=1$. Now, by using the same procedure than previously discussed in the proof of Theorem~\ref{lemma2}, e.g., finding upper and lower bounds for \eqref{Ica}, and averaging both, 
we attain
\begin{align}
\mathcal{L}_{I_c}(s)\approx\sum_{t\in\{1,2\}}\beta_{t-1}\exp\bigg(-\chi\Big(c_1+c_2\big(\frac{t}{2}\big)^{\frac{\alpha-2}{\alpha}}\big(a_i^{\frac{2}{\alpha}}+b_i^{\frac{2}{\alpha}}\big)\Big)s^{\frac{2}{\alpha}}\bigg),\label{aibi0}
\end{align}
where $\beta_0,\beta_1\in[0,1]$ and $\beta_0+\beta_1\!=\!1$.
Since $\tfrac{2}{\alpha}<1$, and the relation between the generalized mean (or power mean) with exponents $\tfrac{2}{\alpha}$ and $1$, the following result holds
\begin{align}
\bigg(\frac{a_i^{\frac{2}{\alpha}}+b_i^{\frac{2}{\alpha}}}{2}\bigg)^{\frac{\alpha}{2}}\le\frac{a_i+b_i}{2}\nonumber\\
a_i^{\frac{2}{\alpha}}+b_i^{\frac{2}{\alpha}}\le2\Big(\frac{\delta}{2}\Big)^{\frac{2}{\alpha}}.\label{aibi}
\end{align}
By substituting \eqref{aibi} into \eqref{aibi0} we attain \eqref{lm7}. See Fig.~\ref{Fig3} for more insights on the accuracy. \hfill 	\qedsymbol

\section{Proof of Theorem~\ref{the6}}\label{App_G}	
In order to attain a fair coexistence between the OMA and NOMA setups the interference must be kept the same. Thus, we need to match \eqref{lm7} with the Laplace transform of the interference for an equivalent OMA setup as shown next.
\begin{align}
\sum_{t\in\{1,2\}}\beta_{t-1}\exp\Big(-\chi(c_1+c_2t^{\frac{\alpha-2}{\alpha}}\delta^{\frac{2}{\alpha}})s^{\frac{2}{\alpha}}\Big)&=\exp\big(-\chi(1-c_0)s^{\frac{2}{\alpha}}\big)\nonumber\\
\exp(-\chi c_1s^{\frac{2}{\alpha}})\Big(\exp(-\chi c_2 (s\delta)^{\frac{2}{\alpha}})+\exp(-\chi c_2 2^{\frac{\alpha-2}{2}}(s\delta)^{\frac{2}{\alpha}})\Big)&\stackrel{(a)}{=}2\exp(-\chi(1-c_0)s^{\frac{2}{\alpha}})\nonumber\\
\exp(-\chi c_2 (s\delta)^{\frac{2}{\alpha}})+\exp(-\chi c_2 2^{\frac{\alpha-2}{2}}(s\delta)^{\frac{2}{\alpha}})&\stackrel{(b)}{=}2\exp(-\chi(1-c_0-c_1)s^{\frac{2}{\alpha}})\nonumber\\
\xi^{\delta^{\frac{2}{\alpha}}}+\xi^{2^{\frac{\alpha-2}{\alpha}}\delta^{\frac{2}{\alpha}}}&\stackrel{(c)}{=}2\xi
\end{align}
where $(a)$ comes from setting $\beta_0=\beta_1=0.5$ and evaluating the left-hand sum,  $(b)$ comes from dividing both terms by $\exp(-\chi c_1 s^{2/\alpha})$, and $(c)$ by $1-c_0-c_1=c_2$ and setting $\xi=\exp(-\chi c_2 s^{2/\alpha})$. Dividing both terms by $\xi$ we attain \eqref{eq}. The solution is an approximation since \eqref{lm7} is an approximation. 
Now, bounding the solution is simple since \eqref{lm7} is a combination of upper, $t=1$, and lower, $t=2$, bounds. Thus,

\begin{tabular}{ C{7.5cm} | C{7.5cm} }
			\vbox{$\begin{aligned}
			\!\exp(-\chi						(c_1\!+\!c_2\delta^{\frac{2}{\alpha}})s^{\frac{2}{\alpha}})&\!\ge\!
			\exp(\!-\chi(1\!-\!c_0)s^{\frac{2}{\alpha}})\nonumber\\
			c_1+c_2\delta^{\frac{2}{\alpha}}&\le1-c_0\nonumber\\
			c_2\delta^{\frac{2}{\alpha}}&\le c_2\nonumber\\
			\delta&\le 1,
			\end{aligned}$} &
			\vbox{$\begin{aligned}
			\!\exp(\!-\!\chi (c_1\!+\!c_2 2^{\frac{\alpha\!-\!2}{\alpha}}\delta^{\frac{2}{\alpha}})s^{\frac{2}{\alpha}})&\!\le\! \exp(\!-\!\chi(1\!-\!c_0)s^{\frac{2}{\alpha}})\nonumber\\
    		c_1+c_2 2^{\frac{\alpha-2}{\alpha}}\delta^{\frac{2}{\alpha}}&\ge 1-c_0\nonumber\\
    		2^{\frac{\alpha-2}{\alpha}}\delta^{\frac{2}{\alpha}}&\ge 1\nonumber\\
	    	\delta&\ge 2^{\frac{2-\alpha}{2}},
	    	\end{aligned}$}
\end{tabular}

completing the proof. \hfill 	\qedsymbol

\bibliographystyle{IEEEtran}
\bibliography{IEEEabrv,references}

\begin{thebibliography}{10}
\providecommand{\url}[1]{#1}
\csname url@samestyle\endcsname
\providecommand{\newblock}{\relax}
\providecommand{\bibinfo}[2]{#2}
\providecommand{\BIBentrySTDinterwordspacing}{\spaceskip=0pt\relax}
\providecommand{\BIBentryALTinterwordstretchfactor}{4}
\providecommand{\BIBentryALTinterwordspacing}{\spaceskip=\fontdimen2\font plus
\BIBentryALTinterwordstretchfactor\fontdimen3\font minus
  \fontdimen4\font\relax}
\providecommand{\BIBforeignlanguage}[2]{{%
\expandafter\ifx\csname l@#1\endcsname\relax
\typeout{** WARNING: IEEEtran.bst: No hyphenation pattern has been}%
\typeout{** loaded for the language `#1'. Using the pattern for}%
\typeout{** the default language instead.}%
\else
\language=\csname l@#1\endcsname
\fi
#2}}
\providecommand{\BIBdecl}{\relax}
\BIBdecl

\bibitem{ETSI1}
{3GPP TS 36.331 V10.50.0}, ``Evolved universal terrestrial radio access
  ({E-UTRA}); radio resource control ({RRC}),'' 2012.

\bibitem{Lin.2014}
T.~M. Lin, C.~H. Lee, J.~P. Cheng, and W.~T. Chen, ``{PRADA}: Prioritized
  random access with dynamic access barring for {MTC} in {3GPP LTE-A
  }networks,'' \emph{IEEE Transactions on Vehicular Technology}, vol.~63,
  no.~5, pp. 2467--2472, Jun 2014.

\bibitem{Yang.2012}
X.~Yang, A.~Fapojuwo, and E.~Egbogah, ``Performance analysis and parameter
  optimization of random access backoff algorithm in {LTE},'' in \emph{IEEE
  Vehicular Technology Conference (VTC Fall)}, Sept 2012, pp. 1--5.

\bibitem{Hossain.2016}
M.~I. Hossain, A.~Azari, and J.~Zander, ``{DERA: Augmented} random access for
  cellular networks with dense {H2H-MTC} mixed traffic,'' in \emph{2016 IEEE
  Globecom Workshops (GC Wkshps)}, Dec 2016, pp. 1--7.

\bibitem{Laya.2016}
A.~Laya, C.~Kalalas, F.~Vazquez-Gallego, L.~Alonso, and J.~Alonso-Zarate,
  ``Goodbye, {ALOHA}!'' \emph{IEEE Access}, vol.~4, pp. 2029--2044, 2016.

\bibitem{Nardelli2015}
P.~H.~J. Nardelli, M.~D.~C. Tom{\'{e}}, H.~Alves, C.~H.~M. de~Lima, and
  M.~Latva-aho, ``{Maximizing the Link Throughput between Smart-meters and
  Aggregators as Secondary Users under Power and Outage Constraints},''
  \emph{Ad Hoc Networks}, 2015.

\bibitem{Dawy.2017}
Z.~Dawy, W.~Saad, A.~Ghosh, J.~G. Andrews, and E.~Yaacoub, ``Toward massive
  machine type cellular communications,'' \emph{IEEE Wireless Communications},
  vol.~24, no.~1, pp. 120--128, February 2017.

\bibitem{Kim.2017}
D.~M. Kim, R.~B. Sorensen, K.~Mahmood, O.~N. Osterbo, A.~Zanella, and
  P.~Popovski, ``Data aggregation and packet bundling of uplink small packets
  for monitoring applications in {LTE},'' \emph{IEEE Network}, vol.~31, no.~6,
  pp. 32--38, November 2017.

\bibitem{Guo.2017}
J.~Guo, S.~Durrani, X.~Zhou, and H.~Yanikomeroglu, ``Massive machine type
  communication with data aggregation and resource scheduling,'' \emph{IEEE
  Transactions on Communications}, vol.~65, no.~9, pp. 4012--4026, Sept 2017.

\bibitem{Kwon.2013}
T.~Kwon and J.~M. Cioffi, ``Random deployment of data collectors for serving
  randomly-located sensors,'' \emph{IEEE Transactions on Wireless
  Communications}, vol.~12, no.~6, pp. 2556--2565, June 2013.

\bibitem{Malak.2016}
D.~Malak, H.~S. Dhillon, and J.~G. Andrews, ``Optimizing data aggregation for
  uplink machine-to-machine communication networks,'' \emph{IEEE Transactions
  on Communications}, vol.~64, no.~3, pp. 1274--1290, March 2016.

\bibitem{Khoshkholgh.2015}
M.~G. Khoshkholgh, Y.~Zhang, K.~G. Shin, V.~C.~M. Leung, and S.~Gjessing,
  ``Modeling and characterization of transmission energy consumption in
  machine-to-machine networks,'' in \emph{IEEE Wireless Communications and
  Networking Conference (WCNC)}, March 2015, pp. 2073--2078.

\bibitem{Chang.2012}
C.-H. Chang and H.-Y. Hsieh, ``Not every bit counts: A resource allocation
  problem for data gathering in machine-to-machine communications,'' in
  \emph{2012 IEEE Global Communications Conference (GLOBECOM)}, Dec 2012, pp.
  5537--5543.

\bibitem{Gotsis.2012}
A.~G. Gotsis, A.~S. Lioumpas, and A.~Alexiou, ``{Evolution of packet scheduling
  for Machine-Type communications over {LTE}: Algorithmic design and
  performance analysis},'' in \emph{IEEE Globecom Workshops}, Dec 2012, pp.
  1620--1625.

\bibitem{Hamdoun.2015}
S.~Hamdoun, A.~Rachedi, and Y.~Ghamri-Doudane, ``Radio resource sharing for
  {MTC} in {LTE-A}: An interference-aware bipartite graph approach,'' in
  \emph{2015 IEEE Global Communications Conference (GLOBECOM)}, Dec 2015, pp.
  1--7.

\bibitem{Kumar.2016}
A.~Kumar, A.~Abdelhadi, and C.~Clancy, ``A delay optimal {MAC} and packet
  scheduler for heterogeneous {M2M} uplink,'' \emph{arXiv preprint
  arXiv:1606.06692}, 2016.

\bibitem{Shirvanimoghaddam.2016}
M.~Shirvanimoghaddam, M.~Dohler, and S.~J. Johnson, ``{Massive Non-Orthogonal
  Multiple Access for Cellular IoT: Potentials and Limitations},'' \emph{IEEE
  Communications Magazine}, vol.~55, no.~9, pp. 55--61, 2017.

\bibitem{Miao.2016}
G.~Miao, A.~Azari, and T.~Hwang, ``{$E^{2}-$MAC}: Energy efficient medium
  access for massive {M2M} communications,'' \emph{IEEE Transactions on
  Communications}, vol.~64, no.~11, pp. 4720--4735, Nov 2016.

\bibitem{Shirvanimoghaddam.2017_2}
M.~Shirvanimoghaddam, M.~Condoluci, M.~Dohler, and S.~J. Johnson, ``On the
  fundamental limits of {Random Non-Orthogonal Multiple Access} in cellular
  massive {IoT},'' \emph{IEEE Journal on Selected Areas in Communications},
  vol.~35, no.~10, pp. 2238--2252, Oct 2017.

\bibitem{Saito.2013_2}
Y.~Saito, A.~Benjebbour, Y.~Kishiyama, and T.~Nakamura, ``System-level
  performance evaluation of downlink non-orthogonal multiple access {(NOMA)},''
  in \emph{IEEE Annual International Symposium on Personal, Indoor, and Mobile
  Radio Communications (PIMRC)}, Sept 2013, pp. 611--615.

\bibitem{Ding.2014}
Z.~Ding, Z.~Yang, P.~Fan, and H.~V. Poor, ``On the performance of
  non-orthogonal multiple access in {5G} systems with randomly deployed
  users,'' \emph{IEEE Signal Processing Letters}, vol.~21, no.~12, pp.
  1501--1505, Dec 2014.

\bibitem{Yang.2016}
Z.~Yang, Z.~Ding, P.~Fan, and N.~Al-Dhahir, ``A general power allocation scheme
  to guarantee quality of service in downlink and uplink {NOMA} systems,''
  \emph{IEEE Transactions on Wireless Communications}, vol.~15, no.~11, pp.
  7244--7257, Nov 2016.

\bibitem{Imari.2014}
M.~Al-Imari, P.~Xiao, M.~A. Imran, and R.~Tafazolli, ``Uplink non-orthogonal
  multiple access for {5G} wireless networks,'' in \emph{International
  Symposium on Wireless Communications Systems (ISWCS)}, Aug 2014, pp.
  781--785.

\bibitem{Abbas.2017}
R.~Abbas, M.~Shirvanimoghaddam, Y.~Li, and B.~Vucetic, ``Grant-free massive
  {NOMA}: {Outage} probability and throughput,'' \emph{arXiv preprint
  arXiv:1707.07401}, 2017.

\bibitem{Ding.2017}
Z.~Ding, Y.~Liu, J.~Choi, Q.~Sun, M.~Elkashlan, C.~L. I, and H.~V. Poor,
  ``Application of non-orthogonal multiple access in {LTE} and {5G} networks,''
  \emph{IEEE Communications Magazine}, vol.~55, no.~2, pp. 185--191, February
  2017.

\bibitem{Zhang.2016}
Z.~Zhang, H.~Sun, R.~Q. Hu, and Y.~Qian, ``Stochastic geometry based
  performance study on {5G} non-orthogonal multiple access scheme,'' in
  \emph{IEEE Global Communications Conference (GLOBECOM)}, Dec 2016, pp. 1--6.

\bibitem{Zhang.2017}
Z.~Zhang, H.~Sun, and R.~Q. Hu, ``Downlink and uplink non-orthogonal multiple
  access in a dense wireless network,'' \emph{IEEE Journal on Selected Areas in
  Communications}, vol.~PP, no.~99, pp. 1--1, 2017.

\bibitem{DLMF}
\BIBentryALTinterwordspacing
``{\it NIST Digital Library of Mathematical Functions},''
  http://dlmf.nist.gov/, Release 1.0.15 of 2017-06-01, {F.~W.~J. Olver, A.~B.
  {Olde Daalhuis}, D.~W. Lozier, B.~I. Schneider, R.~F. Boisvert, C.~W. Clark,
  B.~R. Miller and B.~V. Saunders, eds.} [Online]. Available:
  \url{http://dlmf.nist.gov/}
\BIBentrySTDinterwordspacing

\bibitem{Abramowitz.1972}
M.~Abramowitz, I.~A. Stegun \emph{et~al.}, \emph{Handbook of mathematical
  functions with formulas, graphs, and mathematical tables}.\hskip 1em plus
  0.5em minus 0.4em\relax Dover, New York, 1972, vol.~9.

\bibitem{Liu.2016}
F.~Liu, P.~M\"ah\"onen, and M.~Petrova, ``Proportional fairness-based power
  allocation and user set selection for downlink noma systems,'' in \emph{2016
  IEEE International Conference on Communications (ICC)}, May 2016, pp. 1--6.

\bibitem{Gharbieh.2017}
M.~Gharbieh, H.~ElSawy, A.~Bader, and M.~S. Alouini, ``Spatiotemporal
  stochastic modeling of {IoT} enabled cellular networks: Scalability and
  stability analysis,'' \emph{IEEE Transactions on Communications}, vol.~PP,
  no.~99, pp. 1--1, 2017.

\bibitem{Shirvanimoghaddam.2017}
M.~Shirvanimoghaddam, M.~Dohler, and S.~J. Johnson, ``Massive multiple access
  based on superposition raptor codes for cellular {M2M} communications,''
  \emph{IEEE Transactions on Wireless Communications}, vol.~16, no.~1, pp.
  307--319, Jan 2017.

\bibitem{Azari.2016}
\BIBentryALTinterwordspacing
A.~Azari, ``Energy-efficient scheduling and grouping for machine-type
  communications over cellular networks,'' \emph{Ad Hoc Networks}, vol.~43, pp.
  16--29, 2016, smart Wireless Access Networks and Systems for Smart Cities.
  [Online]. Available:
  \url{http://www.sciencedirect.com/science/article/pii/S1570870516300282}
\BIBentrySTDinterwordspacing

\bibitem{Sun.2016}
H.~Sun, B.~Xie, R.~Q. Hu, and G.~Wu, ``Non-orthogonal multiple access with
  {SIC} error propagation in downlink wireless {MIMO} networks,'' in \emph{IEEE
  Vehicular Technology Conference (VTC-Fall)}, Sept 2016, pp. 1--5.

\bibitem{Haenggi.2012}
M.~Haenggi, \emph{Stochastic geometry for wireless networks}.\hskip 1em plus
  0.5em minus 0.4em\relax Cambridge University Press, 2012.

\bibitem{Haenggi.2014}
------, ``The mean interference-to-signal ratio and its key role in cellular
  and amorphous networks,'' \emph{IEEE Wireless Communications Letters},
  vol.~3, no.~6, pp. 597--600, Dec 2014.

\bibitem{Guo.2015}
A.~Guo and M.~Haenggi, ``Asymptotic deployment gain: A simple approach to
  characterize the {SINR} distribution in general cellular networks,''
  \emph{IEEE Transactions on Communications}, vol.~63, no.~3, pp. 962--976,
  March 2015.

\bibitem{Ganti.2016}
R.~K. Ganti and M.~Haenggi, ``Asymptotics and approximation of the {SIR}
  distribution in general cellular networks,'' \emph{IEEE Transactions on
  Wireless Communications}, vol.~15, no.~3, pp. 2130--2143, March 2016.

\bibitem{Herbert.2006}
H.~A. David and H.~Nagaraja, ``Order statistics,'' \emph{Encyclopedia of
  Statistical Sciences}, 2006.

\bibitem{Renzo.2014}
M.~D. Renzo and P.~Guan, ``Stochastic geometry modeling of coverage and rate of
  cellular networks using the {Gil-Pelaez} inversion theorem,'' \emph{IEEE
  Communications Letters}, vol.~18, no.~9, pp. 1575--1578, Sept 2014.

\end{thebibliography}
\end{document}